\documentclass[default,iicol]{sn-jnl}

\usepackage{verbatim,bm}

\jyear{2023}%

\theoremstyle{thmstyleone}%
\theoremstyle{thmstyletwo}%

\theoremstyle{thmstylethree}%

\theoremstyle{definition}
\newtheorem{dfn}{Definition}[section]
\newtheorem{pro}[dfn]{Problem}
\newtheorem{thm}[dfn]{Theorem}

\newtheorem{cor}[dfn]{Corollary}

\newtheorem{exa}[dfn]{Example}

\renewcommand{\R}{\mathbb R}
\renewcommand{\Z}{\mathbb Z}

\newcommand{\La}{\Lambda}
\newcommand{\ti}{\tilde}

\newcommand{\vol}{\mathrm{Vol}}
\newcommand{\bs}{\hfill $\blacksquare$}

\newcommand{\vl}{\,:\,}

\begin{document}

\title[Density functions of periodic sequences]{Density functions of periodic sequences of continuous events}


\author[1]{\fnm{Olga} \sur{Anosova}}\email{oanosova@liverpool.ac.uk}

\author*[1]{\fnm{Vitaliy} \sur{Kurlin}}\email{vitaliy.kurlin@gmail.com}

\affil*[1]{\orgdiv{Computer Science}, \orgname{University of Liverpool}, \orgaddress{\street{Ashton street}, \city{Liverpool}, \postcode{L69 3BX}, \country{UK}}}


\abstract{
Periodic Geometry studies isometry invariants of periodic point sets that are also continuous under perturbations.
The motivations come from periodic crystals whose structures are determined in a rigid form but any minimal cells can discontinuously change due to small noise in measurements. 
For any integer $k\geq 0$, the density function of a periodic set $S$ was previously defined as the fractional volume of all $k$-fold intersections (within a minimal cell) of balls that have a variable radius $t$ and centers at all points of $S$. 
This paper introduces the density functions for periodic sets of points with different initial radii motivated by atomic radii of chemical elements and by continuous events occupying disjoint intervals in time series.
The contributions are explicit descriptions of the densities for periodic sequences of intervals. 
The new densities are strictly stronger and distinguish periodic sequences that have identical densities in the case of zero radii.
}



\maketitle

\section{Motivations for the density functions of periodic sets}

This work substantially extends the previous conference paper \cite{anosova2022density} in Discrete Geometry and Mathematical Morphology 2022.
The past work explicitly described the density functions for periodic sequences of zero-sized points.
The new work extends these analytic descriptions to periodic sequences whose points have non-negative radii.
\medskip

The proposed extension to the weighted case is motivated by crystallography and materials chemistry \cite{anosova2021introduction} because all chemical elements have different atomic radii.
In dimension 1, the key motivation is the study of periodic time series consisting of continuous and sequential (non-overlapping) events represented by disjoint intervals.
Any such interval $[a,b]\subset\R$ for $a\leq b$ is the one-dimensional ball with the center $\dfrac{a+b}{2}$ and radius $\dfrac{b-a}{2}$.   
\medskip

The point-set representation of periodic crystals is the most fundamental mathematical model for crystalline materials because nuclei of atoms are well-defined physical objects, while chemical bonds are not real sticks or strings but abstractly represent inter-atomic interactions depending on many thresholds for distances and angles.
\medskip

Since crystal structures are determined in a rigid form, their most practical equivalence is \emph{rigid motion} (a composition of translations and rotations) or \emph{isometry} that maintains all inter-point distances and includes also mirror reflections \cite{widdowson2022average}.
\medskip

Now we introduce the key concepts. 
Let $\R^n$ be Euclidean space, $\Z$ be the set of all integers.

\newcommand{\cheight}{32mm}
\newcommand{\bwidth}{66mm}
\newcommand{\dwidth}{92mm}
\begin{figure*}[h!]
\parbox{\bwidth}{
  \includegraphics[height=\cheight]{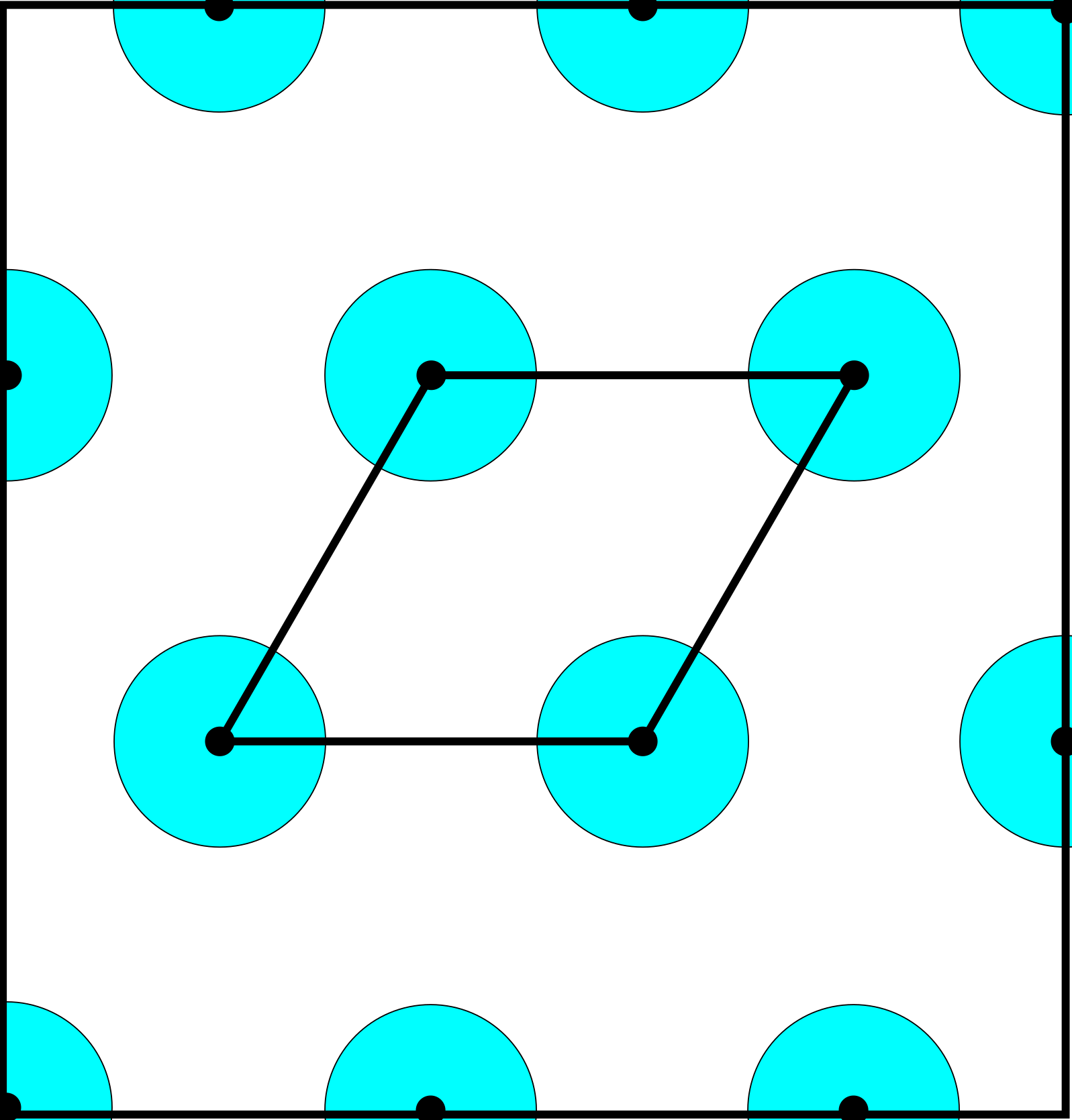}
  \hspace*{1mm}
  \includegraphics[height=\cheight]{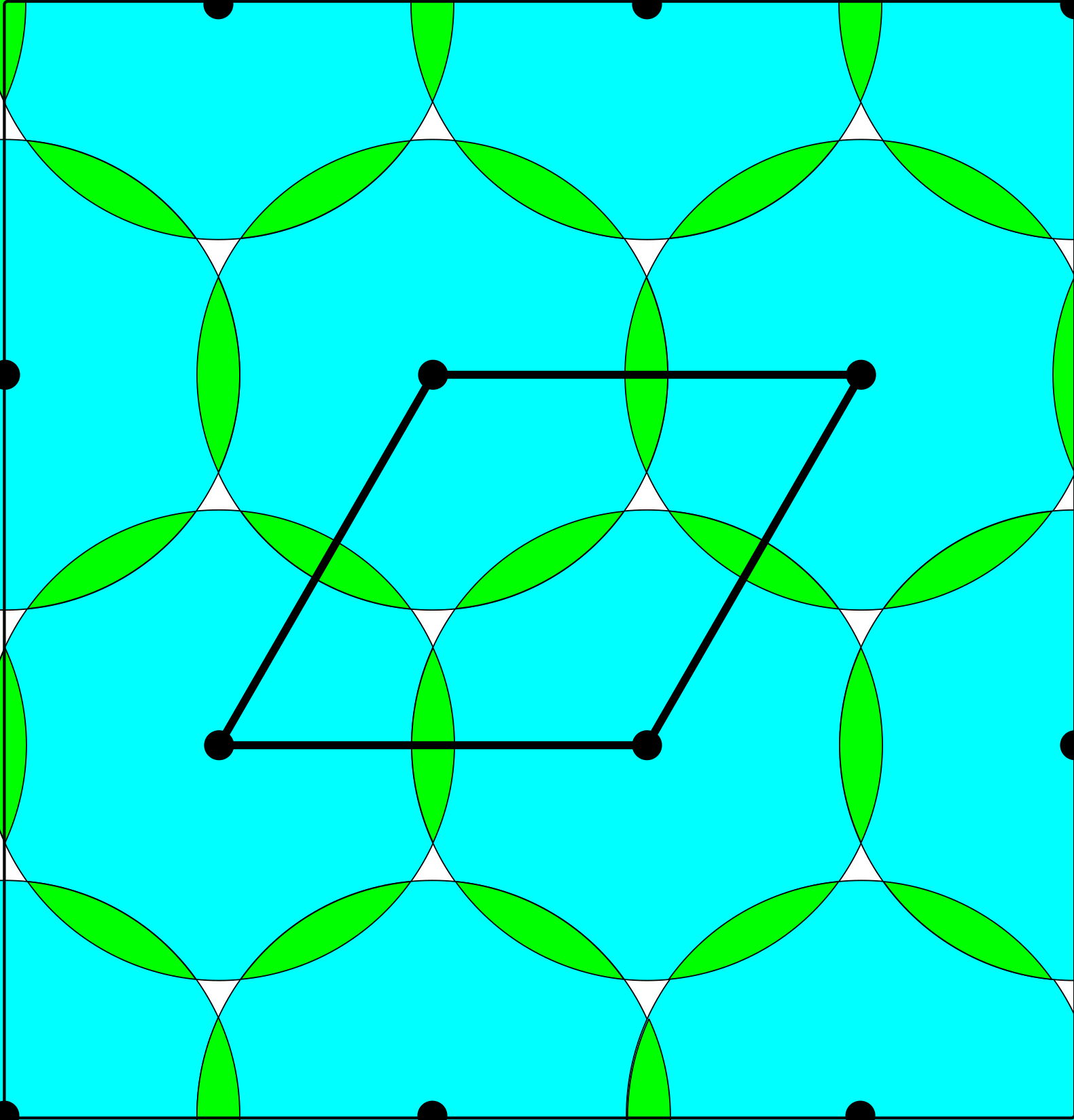}
  \medskip
  
  \includegraphics[height=\cheight]{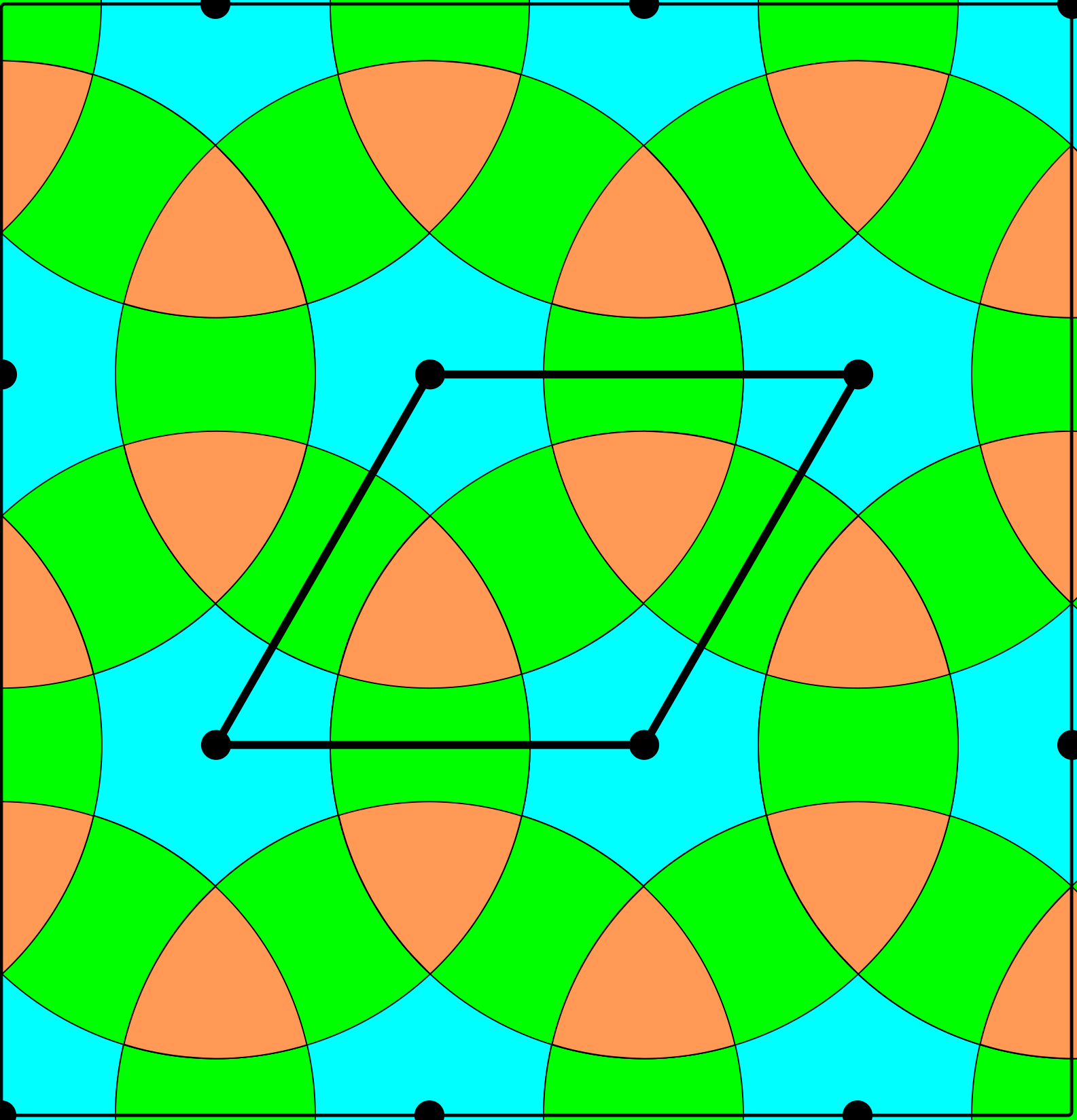}
  \hspace*{1mm}
  \includegraphics[height=\cheight]{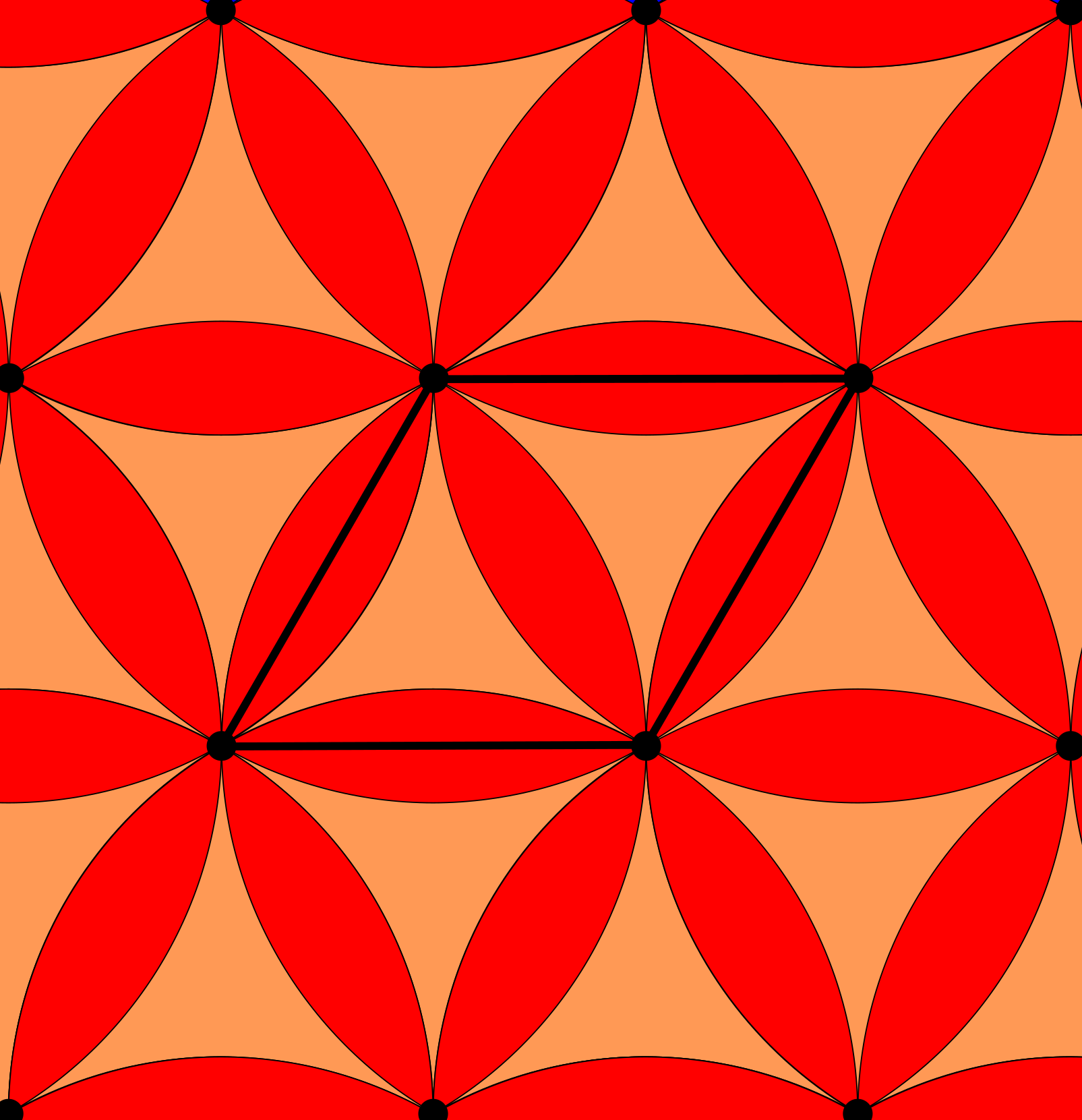}}
  \parbox{80mm}{
  \includegraphics[width=\dwidth]{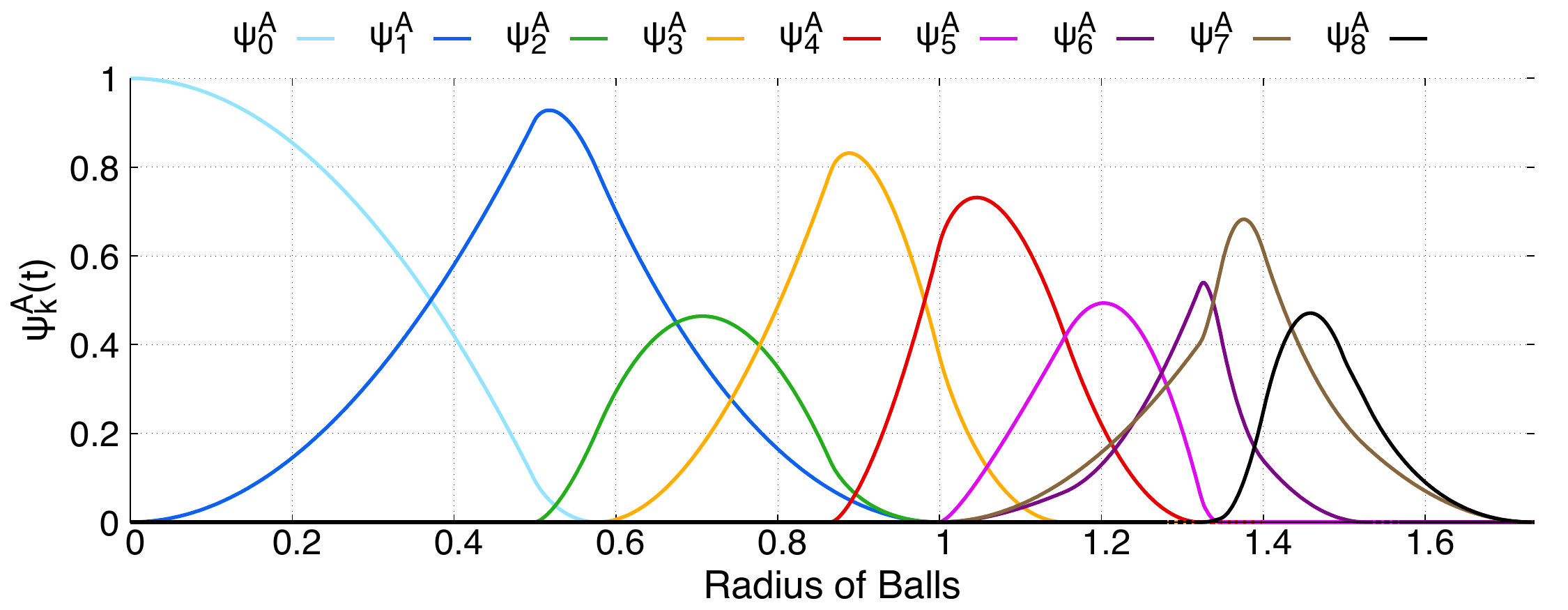}
  \includegraphics[width=\dwidth]{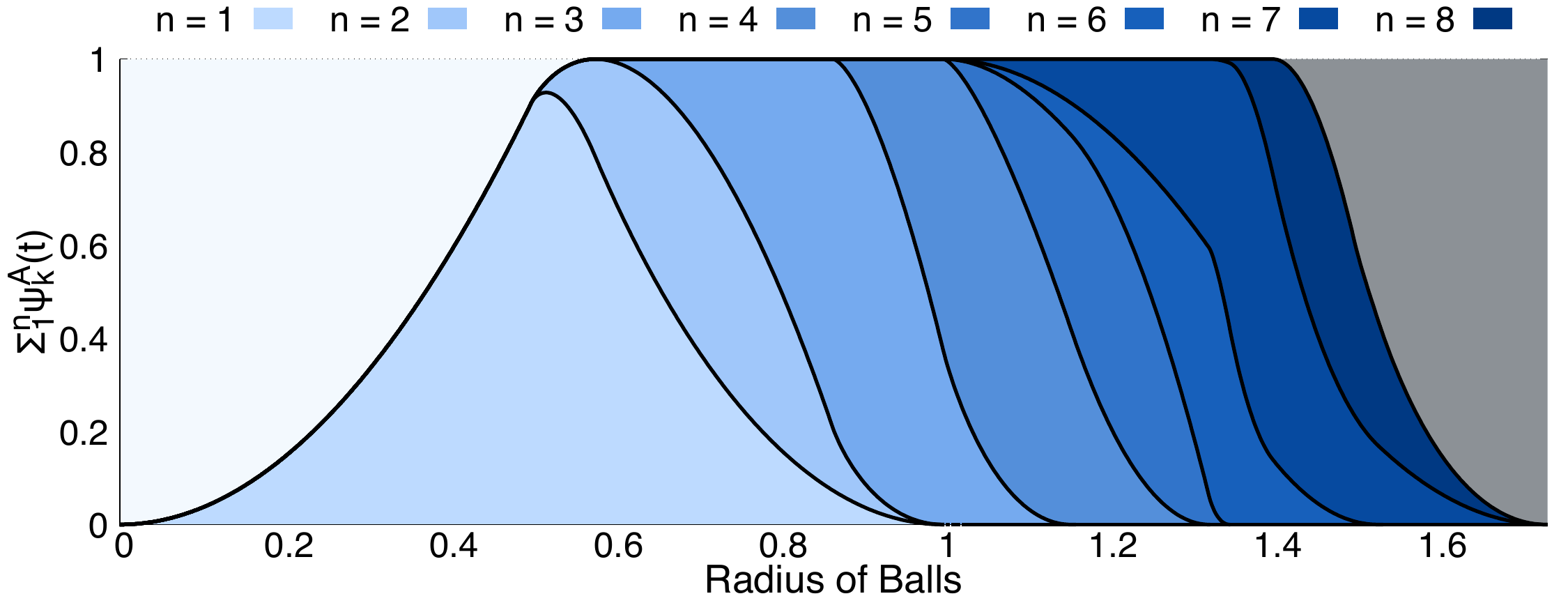}}
  \caption{Illustration of Definition~\ref{dfn:densities} for the hexagonal lattice.
  \textbf{Left}: subregions $U_k(t)$ are covered by $k$ disks for the radii $t=0.25, 0.55, 0.75, 1$.
  \textbf{Right}: the densities $\psi_k$ are above the corresponding \emph{densigram} of 
  accumulated functions 
  $\sum\limits_{i=1}^k\psi_i(t)$.}
  \label{fig:densities_sq}
\end{figure*}

\begin{dfn}[a \emph{lattice} $\La$, a \emph{unit cell} $U$, a \emph{motif} $M$, a \emph{periodic point set} $S=M+\La$]
\label{dfn:crystal}
For any linear basis $ v_1,\dots, v_n$ of $\R^n$, a {\em lattice} is $\La=\{\sum\limits_{i=1}^n c_i v_i : c_i\in\Z\}$.
The \emph{unit cell} $U( v_1,\dots, v_n)=\left\{ \sum\limits_{i=1}^n c_i v_i \vl c_i\in[0,1) \right\}$ is the parallelepiped spanned by the basis above.
A \emph{motif} $M\subset U$ is any finite set of points $p_1,\dots,p_m\in U$.
A \emph{periodic point set} \cite{widdowson2022average}
is the Minkowski sum $S=M+\La=\{ u+ v \mid  u\in M,  v\in \La\}$.
\bs
\end{dfn}

In dimension $n=1$, a lattice is defined by any non-zero vector $v\in\R$, any periodic point set $S$ is a periodic sequence $\{p_1,\dots,p_m\}+|v|\Z$ with the period $|v|$ equal to the length of the vector $v$.

\begin{dfn}[density functions for periodic sets of points with radii]
\label{dfn:densities}
Let a periodic set $S=\La+M\subset\R^n$ have a unit cell $U$.
For every point $p\in M$, fix a radius $r(p)\geq 0$.
For any integer $k\geq 0$, let $U_k(t)$ be the region within the cell $U$ covered by exactly $k$ closed balls $\bar B(p;r(p)+t)$ for $t\geq 0$ and all  points $p\in M$ and their translations by $\La$.
The $k$-th \emph{density} function $\psi_k[S](t)=\vol[U_k(t)]/\vol[U]$ is the fractional volume of the $k$-fold intersections of these balls within $U$.
\bs
\end{dfn}

The density $\psi_k[S](t)$ can be interpreted as the probability that a random (uniformly chosen in $U$) point $q$ is at a maximum distance $t$ to exactly $k$ balls with initial radii $r(p)$ and all centers $p\in S$.
\medskip

For $k=0$, the $0$-th density $\psi_0[S](t)$ measures the fractional volume of the empty space not covered by any expanding balls $\bar B(p;r(p)+t)$
\medskip

In the simplest case of radii $r(p)=0$, the infinite sequence $\Psi[S]=\{\psi_k(t)\}_{k=0}^{+\infty}$ was called in \cite[section~3]{edels2021} the \emph{density fingerprint} of a periodic point set $S$.
For $k=1$ and small $t>0$ while all equal-sized balls $\bar B(p;t)$ remain disjoint, the 1st density $\psi_1[S](t)$ increases proportionally to $t^n$ but later reaches a maximum and eventually drops back to $0$ when all points of $\R^n$ are covered of by at least two balls.
See the densities $\psi_k$, $k=0,\dots,8$ for the square and hexagonal lattices  in \cite[Fig.~2]{edels2021}.
\medskip

The original densities helped find a missing crystal in the Cambridge Structural Database, which was accidentally confused with a slight perturbation (measured at a different temperature) of another crystal (polymorph) with the same chemical composition, see \cite[section~7]{edels2021}. 
\medskip

The new weighted case with radii $r(p)\geq 0$ in Definition~\ref{dfn:densities} is even more practically important due to different Van der Waals radii, which are individually defined for all chemical elements.
\medskip

The key advantage of density functions over other isometry invariants of periodic crystals (such as symmetries or conventional representations based on a geometry of a minimal cell) is their continuity under perturbations, see details in section~\ref{sec:review} reviewing the related past work.
\medskip

The only limitation is the infinite size of densities $\psi_k(t)$ due to the unbounded parameters: integer index $k\geq 0$ and continuous radius $t\geq 0$. 
\medskip

We state the following problem in full generality to motivate future work on these densities.

\begin{pro}[computation of $\psi_k$]
\label{pro:densities}
Verify if the density functions $\psi_k[S](t)$ from Definition~\ref{dfn:densities} 
can be computed in a polynomial time (in the size $m$ of a motif of $S$) for a fixed dimension $n$.
\bs 
\end{pro}

The main contribution is the full solution of Problem~\ref{pro:densities} for $n=1$.
Theorems~\ref{thm:0-th_density},~\ref{thm:1st_density},~\ref{thm:k-th_density},~\ref{thm:periodicity}, and Corollary~\ref{cor:computation} efficiently compute all $\psi_k[S](t)$ depending on infinitely many values of $k$ and $t$.

\section{Review of related past work}
\label{sec:review}

Periodic Geometry was initiated in 2020 by the problem \cite[section~2.3]{mosca2020voronoi} to design a computable metric on isometry classes of lattices, which is continuous under perturbations of a lattice basis.  
\medskip

Though a Voronoi domain is combinatorially unstable under perturbations, its geometric shape was used to introduce two continuous metrics \cite[Theorems~2,~4]{mosca2020voronoi} requiring approximations due to a minimization over infinitely many rotations. 
\medskip

Similar minimizations over rotations or other continuous parameters are required for the complete invariant isosets \cite{anosova2021isometry,anosova2022recognition} and density functions, which can be practically computed in low dimensions \cite{smith2022practical}, whose completeness was proved for generic periodic point sets in $\R^3$ \cite[Theorem~2]{edels2021}.
The density fingerprint $\Psi[S]$ turned out to be incomplete \cite[section~5]{edels2021} in the example below.

\begin{exa}[periodic sequences $S_{15},Q_{15}\subset\R$]
\label{exa:SQ15}
Widdowson et al. \cite[Appendix~B]{widdowson2022average} discussed homometric sets that can be distinguished by the invariant AMD (Average Minimum Distances) and not by diffraction patterns.
The sequences
\smallskip

$\begin{array}{l}
S_{15} = \{0,1,3,4,5,7,9,10,12\}+15\Z,\\
Q_{15} = \{0,1,3,4,6,8,9,12,14\}+15\Z
\end{array}$ 
\smallskip

\noindent
have the unit cell $[0,15]$ shown as a circle in Fig.~\ref{fig:SQ15}.
\medskip

\begin{figure}[h!]
\includegraphics[width=\linewidth]{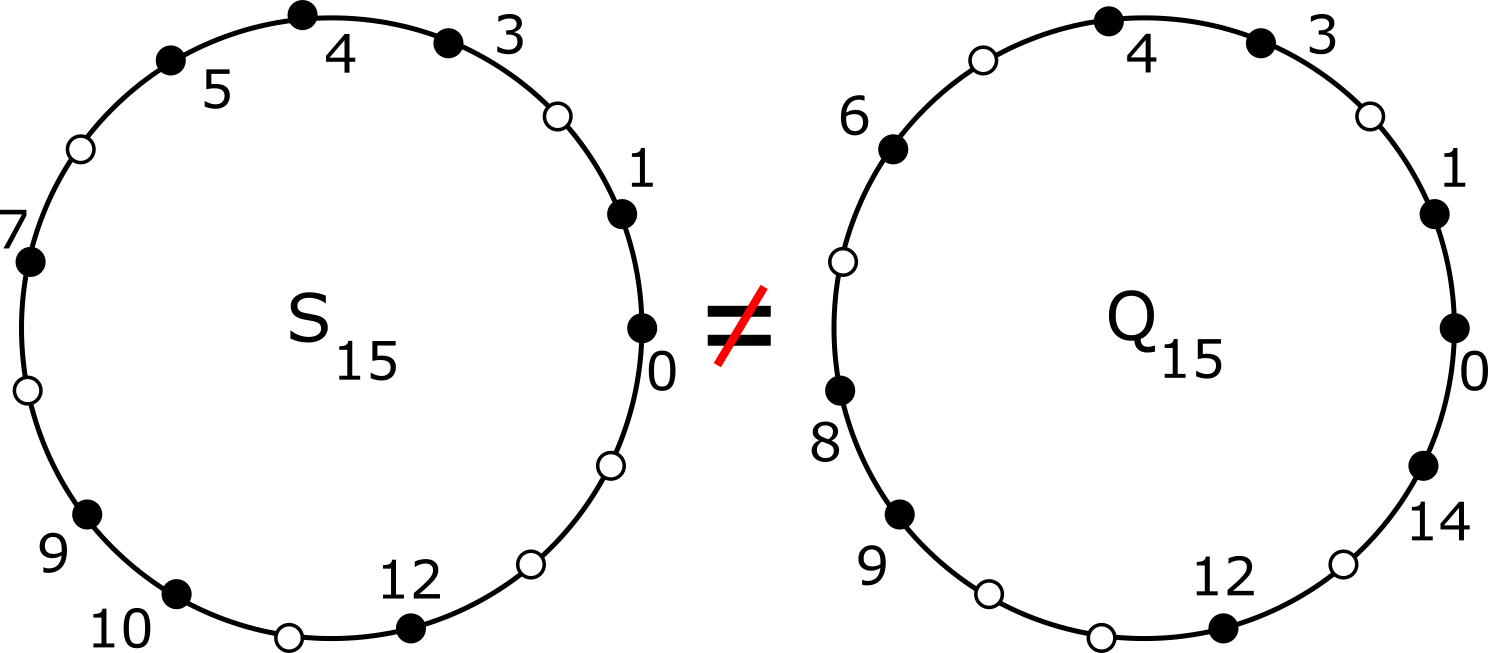}
\caption{Circular versions of the periodic sets $S_{15},Q_{15}$.}
\label{fig:SQ15}      
\end{figure}

These periodic sequences \cite{grunbaum1995use} are obtained as Minkowski sums $S_{15}=U+V+15\Z$ and $Q_{15}=U-V+15\Z$ for 
$U = \{0, 4, 9\}$, $V = \{0, 1, 3\}$.
\bs
\end{exa}

For rational-valued periodic sequences, \cite[Theorem~4]{grunbaum1995use} proved that $r$-th order invariants (combinations of $r$-factor products) up to $r=6$ are enough to distinguish such sequences up to a shift (a rigid motion of $\R$ without reflections).
\medskip

The AMD invariant was extended to the Pointwise Distance Distribution (PDD), whose generic completeness \cite[Theorem~4.4]{widdowson2022resolving} was proved in any dimension $n\geq 1$.
However there are finite sets in $\R^3$ \cite[Fig.~S4]{pozdnyakov2020incompleteness} with the same PDD, which were distinguished by more sophisticated distance-based invariants in \cite[appendix~C]{widdowson2021pointwise}.
\medskip

The subarea of Lattice Geometry developed continuous parameterizations for the moduli spaces of lattices considered up to isometry in dimension two \cite{kurlin2022mathematics,bright2023geographic} and three \cite{kurlin2022complete,bright2021welcome}.
\medskip

For 1-periodic sequences of points in $\R^n$, complete isometry invariants with continuous and computable metrics appeared in \cite{kurlin2022exactly}, see
related results for finite clouds of unlabeled points \cite{smith2022families,kurlin2022computable}.

\section{The 0-th density function $\psi_0$}
\label{sec:0-th_function}

This section proves Theorem~\ref{thm:0-th_density} explicitly describing the 0-th density function $\psi_0[S](t)$ for any periodic sequence $S\subset\R$ of disjoint intervals.
\medskip

For convenience, scale any periodic sequence $S$ to period 1 so that $S$ is given by points $0\leq p_1<\cdots<p_m<1$ with radii $r_1,\dots,r_m$, respectively.
Since the expanding balls in $\R$ are growing intervals, volumes of their intersections linearly change with respect to the variable radius $t$.
Hence any density function $\psi_k(t)$ is piecewise linear and uniquely determined by \emph{corner} points $(a_j,b_j)$ where the gradient of $\psi_k(t)$ changes.
\medskip

To prepare the proof of Theorem~\ref{thm:0-th_density}, we first consider 
Example~\ref{exa:0-th_density} for the simple sequence $S$.

\begin{exa}[$0$-th density function $\psi_0$]
\label{exa:0-th_density}
Let the periodic sequence $S=\left\{0,\dfrac{1}{3},\dfrac{1}{2}\right\}+\Z$ have three points $p_1=0$, $p_2=\dfrac{1}{3}$, $p_3=\dfrac{1}{2}$ of radii $r_1=\dfrac{1}{12}$, $r_2=0$, $r_3=\dfrac{1}{12}$, respectively.
Fig.~\ref{fig:growing_intervals} shows each point $p_i$ and its growing interval 
$$L_i(t)=[(p_i-r_i)-t,(p_i+r_i)+t] \text{ of the length }2r_i+2t$$ 
for $i=1,2,3$ in its own color: red, green, blue.
\medskip

\begin{figure*}[h!]
\includegraphics[width=\textwidth]{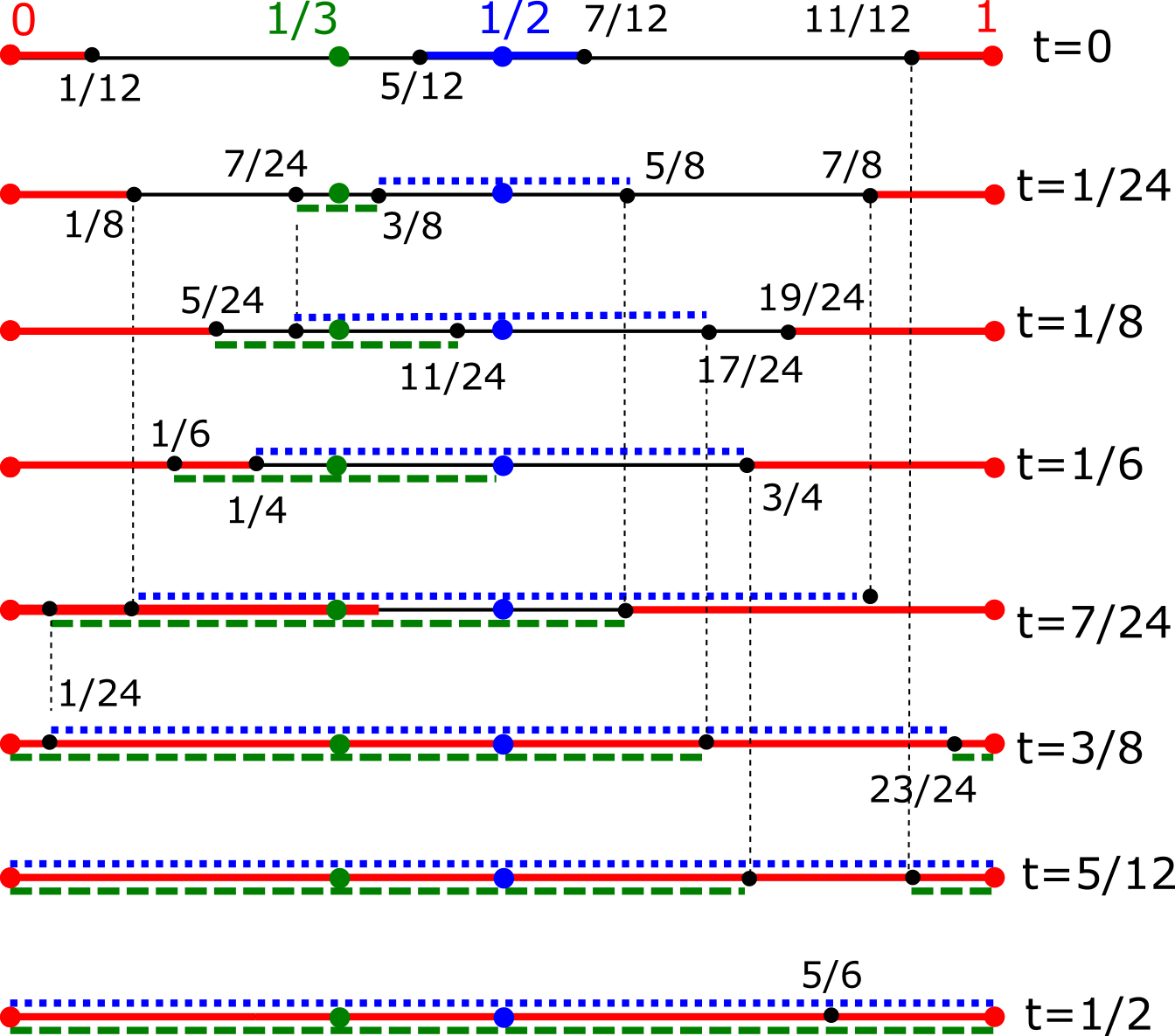}
\caption{The sequence $S=\left\{0,\dfrac{1}{3},\dfrac{1}{2}\right\}+\Z$ has the points of weights $\dfrac{1}{12},0,\dfrac{1}{12}$, respectively. 
The growing intervals around the red point $0\equiv 1\pmod{1}$, green point $\dfrac{1}{3}$, blue point $\dfrac{1}{2}$ have the same color for various radii $t$, see Examples~\ref{exa:0-th_density},~\ref{exa:1st_density},~\ref{exa:2nd_density}.
}
\label{fig:growing_intervals}      
\end{figure*}

By Definition~\ref{dfn:densities} each density function $\psi_k[S](t)$ measures a fractional length covered by exactly $k$ intervals within the unit cell $[0,1]$.
We periodicaly map the endpoints of each growing interval to the unit cell $[0,1]$.
For instance, the interval $[-\dfrac{1}{12}-t,\dfrac{1}{12}+t]$ of the point $p_1=0\equiv 1\pmod{1}$ maps to the red intervals $[0,\dfrac{1}{12}+t]\cup[\dfrac{11}{12}-t,1]$ shown by solid red lines in Fig.~\ref{fig:growing_intervals}.
The same image shows the green interval $[\dfrac{1}{3}-t,\dfrac{1}{3}+t]$ by dashed lines and the blue interval $[\dfrac{5}{12}-t,\dfrac{7}{12}+t]$ by dotted lines.
\medskip

At the moment $t=0$, since the starting intervals are disjoint, they cover the length $l=2(\dfrac{1}{12} + 0 +\dfrac{1}{12})=\dfrac{1}{3}$.
The non-covered part of $[0,1]$ has length $1-\dfrac{1}{3}=\dfrac{2}{3}$.
So the graph of $\psi_0(t)$ at $t=0$ starts from the point $(0,\dfrac{2}{3})$, see Fig.~\ref{fig:0-th_density}.
\medskip

\begin{figure}[h!]
\includegraphics[width=\linewidth]{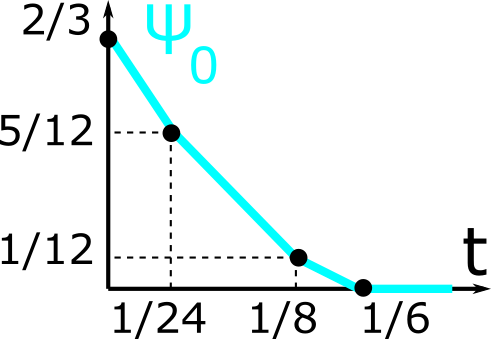}
\caption{The 0-th density function $\psi_0(t)$ for the 1-period sequence $S$  whose points $0,\dfrac{1}{3},\dfrac{1}{2}$ have radii $\dfrac{1}{12},0,\dfrac{1}{12}$, respectively, see Example~\ref{exa:0-th_density}. }
\label{fig:0-th_density}      
\end{figure}

At the first critical moment $t=\dfrac{1}{24}$ when the green and blue intervals collide at $p=\dfrac{3}{8}$, only the intervals $[\dfrac{1}{8},\dfrac{7}{24}]\cup[\dfrac{5}{8},\dfrac{7}{8}]$ of total length $\dfrac{5}{12}$ remain uncovered.
Hence $\psi_0(t)$ linearly drops to the point $(\dfrac{1}{12},\dfrac{5}{12})$. 
At the next critical moment $t=\dfrac{1}{8}$ when the red and green intervals collide at $p=\dfrac{5}{24}$, only the interval $[\dfrac{17}{24},\dfrac{19}{24}]$ of length $\dfrac{1}{12}$ remain uncovered, so $\psi_0(t)$ continues to $(\dfrac{1}{8},\dfrac{1}{12})$.
\medskip

The graph of $\psi_0(t)$ finally returns to the $t$-axis at the point $(\dfrac{1}{6},0)$ and remains there for $t\geq \dfrac{1}{6}$.
\medskip

The piecewise linear behavior of $\psi_0(t)$ can be described by specifying the \emph{corner} points in Fig.~\ref{fig:0-th_density}: 
$\left(0,\dfrac{2}{3}\right)$, 
$\left(\dfrac{1}{24},\dfrac{5}{12}\right)$, 
$\left(\dfrac{1}{8},\dfrac{1}{12}\right)$, 
$\left(\dfrac{1}{6},0\right)$.
\bs
\end{exa}

Theorem~\ref{thm:0-th_density} extends Example~\ref{exa:0-th_density}
 to any periodic sequence $S$ and implies that the 0-th density function $\psi_0(t)$ is uniquely determined by the ordered gap lengths between successive intervals.

\begin{thm}[description of $\psi_0$]
\label{thm:0-th_density}
Let a periodic sequence $S=\{p_1,\dots,p_m\}+\Z$ consist of disjoint intervals with centers $0\leq p_1<\dots<p_m<1$ and radii $r_1,\dots,r_m\geq 0$. 
Consider the \emph{total length} $l=2\sum\limits_{i=1}^m r_i$ and \emph{gaps} between successive intervals $g_i=(p_{i}-r_{i})-(p_{i-1}+r_{i-1})$, where $i=1,\dots,m$ and $p_{0}=p_m-1$, $r_0=r_m$. 
Put the gaps in increasing order: $g_{[1]}\leq g_{[2]}\leq\dots\leq g_{[m]}$.
\medskip

Then the 0-th density $\psi_0[S](t)$ is piecewise linear with the following (unordered) corner points: $(0,1-l)$ and $\left(\dfrac{g_{[i]}}{2},\; 1-l-\sum\limits_{j=1}^{i-1} g_{[j]}-(m-i+1)g_{[i]}\right)$ for $i=1,\dots,m$, so the last corner is $\left(\dfrac{g_{[m]}}{2},0\right)$.
\medskip

If any corners are repeated, e.g. when $g_{[i-1]}=g_{[i]}$, these corners are collapsed into one corner. 
\bs
\end{thm}
\begin{proof}
By Definition~\ref{dfn:densities} the 0-th density function $\psi_0(t)$ measures the total length of subintervals in the unit cell $[0,1]$ that are not covered by any of the growing intervals $L_i(t)=[p_i-r_i-t,p_i+r_i+t]$, $i=1,\dots,m$. 
For $t=0$, since all initial intervals $L_i(0)$ are disjoint, they cover the total length $2\sum\limits_{i=1}^m r_i=l$.
\medskip

Then the graph of $\psi_0(t)$ at $t=0$ starts from the point $(0,1-l)$.
So $\psi_0(t)$ linearly decreases from the initial value $\psi_0(0)=1-l$ except for $m$ critical values of $t$ where one of the gap intervals $[p_i+r_i+t,p_{i+1}-r_{i+1}-t]$ between successive growing intervals $L_i(t)$ and $L_{i+1}(t)$ shrinks to a point.
These critical radii $t$ are ordered according to the gaps $g_{[1]}\leq g_{[2]}\leq\dots\leq g_{[m]}$.
\medskip

The first critical radius is $t=\dfrac{1}{2}g_{[1]}$, when a shortest gap interval of the length $g_{[1]}$ is covered by the growing successive intervals.
At this moment $t=\dfrac{1}{2}g_{[1]}$, all $m$ growing intervals $L_i(t)$ have the total length $l+mg_{[1]}$.
Then the 0-th density $\psi_0(t)$ has the first corner points $(0,1-l)$ and $\left(\dfrac{g_{[1]}}{2},1-l-mg_{[1]}\right)$. 
\medskip

The second critical radius is $t=\dfrac{g_{[2]}}{2}$, when all intervals $L_i(t)$ have the total length $l+g_{[1]}+(m-1)g_{[2]}$, i.e. the next corner point is $\left(\dfrac{g_{[2]}}{2},1-l-g_{[1]}-(m-1)g_{[2]}\right)$. 
If $g_{[1]}=g_{[2]}$, then both corner points coincide, so $\psi_0(t)$ will continue from the joint corner point.
\medskip

The above pattern generalizes to the $i$-th critical radius $t=\dfrac{1}{2}g_{[i]}$, when all covered intervals have the total length $\sum\limits_{j=1}^{i-1}g_{[j]}$ (for the fully covered intervals) plus $(m-i+1)g_{[i]}$ (for the still growing intervals).
\medskip

For the final critical radius $t=\dfrac{g_{[m]}}{2}$, the whole unit cell $[0,1]$ is covered by the grown intervals because $\sum\limits_{j=1}^{m}g_{[j]}=1-l$.
The final corner is $(\dfrac{g_{[m]}}{2},0)$.
\end{proof}

Example~\ref{exa:revisit_0-th_density} applies Theorem~\ref{thm:0-th_density} to get $\psi_0$ found for the periodic sequence $S$ in Example~\ref{exa:0-th_density}.

\begin{exa}[using Theorem~\ref{thm:0-th_density}]
\label{exa:revisit_0-th_density}
The sequence $S=\left\{0,\dfrac{1}{3},\dfrac{1}{2}\right\}+\Z$ in Example~\ref{exa:0-th_density} with points $p_1=0$, $p_2=\dfrac{1}{3}$, $p_3=\dfrac{1}{2}$ of radii $r_1=\dfrac{1}{12}$, $r_2=0$, $r_3=\dfrac{1}{12}$, respectively, has 
$l=2(r_1+r_2+r_3)=\dfrac{1}{3}$ and the initial gaps 
between successive intervals 
$g_1=p_{1}-r_{1}-p_{3}-r_{3}=(1-\dfrac{1}{12})-(\dfrac{1}{2}+\dfrac{1}{12})=\dfrac{1}{3}$, \\ \\
$g_2=p_{2}-r_{2}-p_{1}-r_{1}=(\dfrac{1}{3}-0)-(0+\dfrac{1}{12})=\dfrac{1}{4}$, \\ \\
$g_3=p_{3}-r_{3}-p_{2}-r_{2}=(\dfrac{1}{2}-\dfrac{1}{12})-(\dfrac{1}{3}+0)=\dfrac{1}{12}$.
Order the gaps: 
$g_{[1]}=\dfrac{1}{12}<g_{[2]}=\dfrac{1}{4}<g_{[3]}=\dfrac{1}{3}$.
$1-l=1-\dfrac{1}{3}=\dfrac{2}{3}$, \\ \\
$1-l-3g_{[1]}=\dfrac{2}{3}-\dfrac{3}{12}=\dfrac{5}{12}$, \\ \\
$1-l-g_{[1]}-2g_{[2]}=\dfrac{2}{3}-\dfrac{1}{12}-\dfrac{2}{4}=\dfrac{1}{12}$, \\ \\
$1-l-g_{[1]}-g_{[2]}-g_{[3]}=\dfrac{2}{3}-\dfrac{1}{12}-\dfrac{1}{4}-\dfrac{1}{3}=0$. \\
By Theorem~\ref{thm:0-th_density} $\psi_0(t)$ has the corner points
$(0,1-l)=\left(0,\dfrac{2}{3}\right)$, \\ \\
$\left(\dfrac{1}{2}g_{[1]},1-l-3g_{[1]}\right)=\left(\dfrac{1}{24},\dfrac{5}{12}\right)$, \\ \\
$\left(\dfrac{1}{2}g_{[2]},1-l-g_{[1]}-2g_{[2]}\right)=\left(\dfrac{1}{8},\dfrac{1}{12}\right)$, \\ \\
$\left(\dfrac{1}{2}g_{[3]},1-l-g_{[1]}-g_{[2]}-g_{[3]}\right)=\left(\dfrac{1}{6},0\right)$.
See the graph of the 0-th density $\psi_0(t)$ in Fig.~\ref{fig:0-th_density}.
\bs
\end{exa}

By Theorem~\ref{thm:0-th_density} any 0-th density function $\psi_0(t)$ is uniquely determined by the (unordered) set of gap lengths between successive intervals.
Hence we can re-order these intervals without changing $\psi_0(t)$.
For instance, the periodic sequence $Q=\{0,\dfrac{1}{2},\dfrac{2}{3}\}+\Z$ with points $0,\dfrac{1}{2},\dfrac{2}{3}$ of weights $\dfrac{1}{12},\dfrac{1}{12},0$ has the same set ordered gaps $g_{[1]}=\dfrac{1}{12}$, $d_{[2]}=\dfrac{1}{3}$, $d_{[3]}=\dfrac{1}{2}$ as the periodic sequence $S=\left\{0,\dfrac{1}{3},\dfrac{1}{2}\right\}+\Z$ in Example~\ref{exa:0-th_density}.
\medskip

The above sequences $S,Q$ are related by the mirror reflection $t\mapsto 1-t$.
One can easily construct many non-isometric sequences with $\psi_0[S](t)=\psi_0[Q](t)$.
For any $1\leq i\leq m-3$, the sequences $S_{m,i}=\{0,2,3,\dots,i+2,i+4,i+5,\dots,m+2\}+(m+2)\Z$ have the same interval lengths $d_{[1]}=\dots=d_{[m-2]}=1$, $d_{[m-1]}=d_{[m]}=2$ but are not related by isometry (translations and reflections in $\R$) because the intervals of length 2 are separated by $i-1$ intervals of length 1 in $S_{m,i}$.  

\section{The 1st density function $\psi_1$}
\label{sec:1st_density}

This section proves Theorem~\ref{thm:1st_density} explicitly describing the 1st density $\psi_1[S](t)$ for any periodic sequence $S$ of disjoint intervals.
To prepare the proof of Theorem~\ref{thm:1st_density}, 
Example~\ref{exa:1st_density} finds $\psi_1[S]$ for the sequence $S$ from Example~\ref{exa:0-th_density}.

\begin{figure}[h!]
\includegraphics[width=\linewidth]{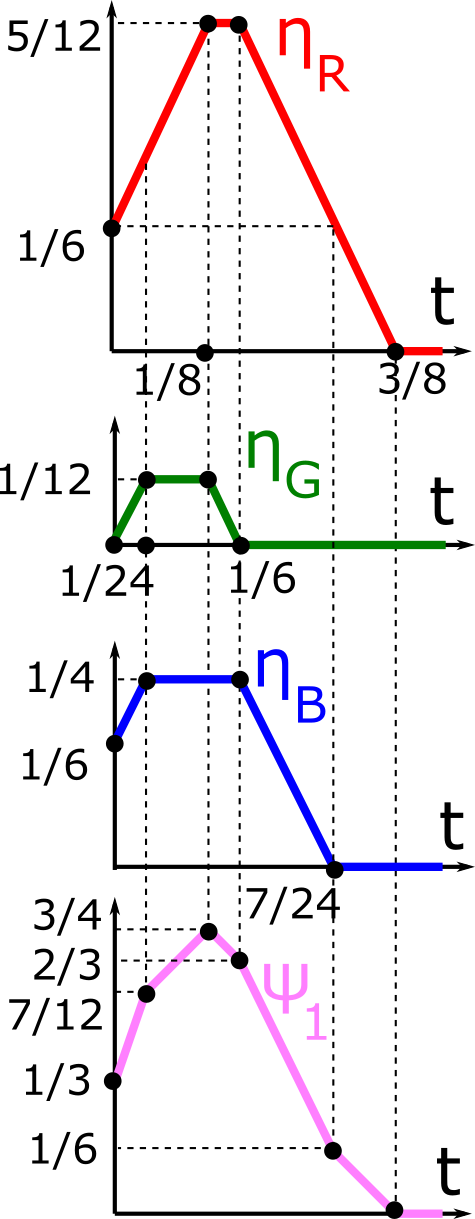}
\caption{The trapezoid functions $\eta_R,\eta_G,\eta_B$ and the 1st density function $\psi_1(t)$ for the 1-period sequence $S$  whose points $0,\dfrac{1}{3},\dfrac{1}{2}$ have radii $\dfrac{1}{12},0,\dfrac{1}{12}$, see Example~\ref{exa:1st_density}. }
\label{fig:1st_density}      
\end{figure}

\begin{exa}[$\psi_1$ for $S=\left\{0,\dfrac{1}{3},\dfrac{1}{2}\right\}+\Z$]
\label{exa:1st_density}
The 1st density function $\psi_1(t)$ can be obtained as a sum of the three \emph{trapezoid} functions $\eta_R$, $\eta_G$, $\eta_B$, each measuring the length of a region covered by a single interval of one color, see Fig.~\ref{fig:growing_intervals}.
\medskip

At the initial moment $t=0$, the red intervals $[0,\dfrac{1}{12}]\cup[\dfrac{11}{12},1]$ have the total length $\eta_R(0)=\dfrac{1}{6}$.
These red intervals $[0,\dfrac{1}{12}+t]\cup[\dfrac{11}{12}-t,1]$ for $t\in[0,\dfrac{1}{8}]$ grow until they touch the green interval $[\dfrac{7}{24},\dfrac{3}{8}]$ and have the total length $\eta_R(\dfrac{1}{8})=\dfrac{1}{6}+\dfrac{2}{8}=\dfrac{5}{12}$ in the second picture of Fig.~\ref{fig:growing_intervals}.
So the graph of the red length $\eta_R(t)$ linearly grows with gradient 2 from the point $(0,\dfrac{1}{6})$ to the corner point $(\dfrac{1}{8},\dfrac{5}{12})$.
\medskip

For $t\in[\dfrac{1}{8},\dfrac{1}{6}]$, the left red interval is shrinking at the same rate (due to the overlapping green interval) as the right red interval continues to grow until $t=\dfrac{1}{6}$, when it touches the blue interval $[\dfrac{1}{4},\dfrac{3}{4}]$.  
Hence the graph of $\eta_R(t)$ remains constant for $t\in[\dfrac{1}{8},\dfrac{1}{6}]$ up to the corner point $(\dfrac{1}{6},\dfrac{5}{12})$.
\medskip

After that, the graph of $\eta_R(t)$ linearly decreases (with gradient $-2$) until all red intervals are fully covered by the green and blue intervals at moment $t=\dfrac{3}{8}$, see the 6th picture in Fig.~\ref{fig:growing_intervals}.
\medskip

Hence the trapezoid function $\eta_R$ has the piecewise linear graph through the corner points $(0,\dfrac{1}{6})$, $(\dfrac{1}{8},\dfrac{5}{12})$, $(\dfrac{1}{6},\dfrac{5}{12})$, $(\dfrac{3}{8},0)$.
After that, $\eta_R(t)=0$ remains constant for $t\geq \dfrac{3}{8}$.
Fig.~\ref{fig:1st_density} shows the graphs of $\eta_R,\eta_G,\eta_B$ and $\psi_1=\eta_R+\eta_G+\eta_B$. 
\bs
\end{exa}

Theorem~\ref{thm:1st_density} extends Example~\ref{exa:1st_density} 
and proves that any $\psi_1(t)$ is a sum of trapezoid functions whose corners are explicitly described. 
We consider any index $i=1,\dots,m$ (of a point $p_i$ or a gap $g_i$) modulo $m$ so that $m+1\equiv 1\pmod{m}$.

\begin{thm}[description of $\psi_1$]
\label{thm:1st_density}
Let a periodic sequence $S=\{p_1,\dots,p_m\}+\Z$ consist of disjoint intervals with centers $0\leq p_1<\dots<p_m<1$ and radii $r_1,\dots,r_m\geq 0$, respectively. 
\medskip

Consider the \emph{gaps} $g_i=(p_{i}-r_{i})-(p_{i-1}+r_{i-1})$, 
where $i=1,\dots,m$ and $p_{0}=p_m-1$, $r_0=r_m$. 
\medskip

Then the 1st density 
$\psi_1(t)$ is the sum of $m$ \emph{trapezoid} functions $\eta_{i}$, $i=1,\dots,m$, with the corners 
$(0,2r_{i})$,
$\left(\dfrac{g_{i}}{2}, g+2r_i\right)$,
$\left(\dfrac{g_{i+1}}{2}, g+2r_i\right)$,  
$\left(\dfrac{g_{i}+g_{i+1}}{2}+r_i,0\right),$ 
where 
$g=\min\{g_{i},g_{i+1}\}$.
\medskip

Hence $\psi_1(t)$ is determined by the unordered set of unordered pairs $(g_{i},g_{i+1})$, $i=1,\dots,m$.
\bs 
\end{thm}
\begin{proof}
The 1st density $\psi_1(t)$ equals the total length of subregions covered by exactly one of the intervals $L_i(t)=[p_i-r_i-t,p_i+r_i+t]$, $i=1,\dots,m$, where all intervals are taken modulo 1 within $[0,1]$.
\medskip

Hence $\psi_1(t)$ is the sum of the functions $\eta_{1i}$, each measuring the length of the subinterval of $L_i(t)$ not covered by other intervals $L_j(t)$, $j\in\{1,\dots,m\}-\{i\}$.
\medskip

Since the initial intervals $L_i(0)$ are disjoint, each function $\eta_{1i}(t)$ starts from the value $\eta_{1i}(0)=2r_i$ and linearly grows (with gradient 2) up to $\eta_{i}(\dfrac{1}{2}g)=2r_i+g$, where $g=\min\{g_{i},g_{i+1}\}$, when the growing interval $L_i(t)$ of the length $2r_i+2t=2r_i+g$ touches its closest neighboring interval $L_{i\pm 1}(t)$ with a shortest gap $g$.
\medskip

If (say) $g_{i}<g_{i+1}$, then the subinterval covered only by $L_i(t)$ is shrinking on the left and is growing at the same rate on the right until $L_i(t)$ touches the growing interval $L_{i+1}(t)$ on the right.
During this growth, when $t$ is between $\dfrac{1}{2}g_{i}$ and $\dfrac{1}{2}g_{i+1}$, the trapezoid function $\eta_{i}(t)=g$ remains constant.
\medskip

If $g_{i}=g_{i+1}$, this horizontal line collapses to one point in the graph of $\eta_{i}(t)$.
For $t\geq\max\{g_{i},g_{i+1}\}$, the subinterval covered only by $L_i(t)$ is shrinking on both sides until the neighboring intervals $L_{i\pm 1}(t)$ meet at a mid-point between their initial closest endpoints $p_{i-1}+r_{i-1}$ and $p_{i+1}-r_{i+1}$.
This meeting time is $t=\dfrac{1}{2}(p_{i+1}-r_{i+1}-p_{i-1}-r_{i-1}) =\dfrac{1}{2}(g_{i}+2r_i+g_{i+1})$, which is also illustrated by Fig.~\ref{fig:intervals_gaps}. 
So the trapezoid function $\eta_{i}$ has the corners 
$(0,2r_i)$, 
$\left(\dfrac{g_{i}}{2}, 2r_i+g\right)$, 
$\left(\dfrac{g_{i+1}}{2},2r_i+g\right)$, 
$\left(\dfrac{g_{i}+g_{i+1}}{2}+r_i,0\right)$
as expected.
\end{proof}

\begin{figure*}[h!]
\includegraphics[width=\linewidth]{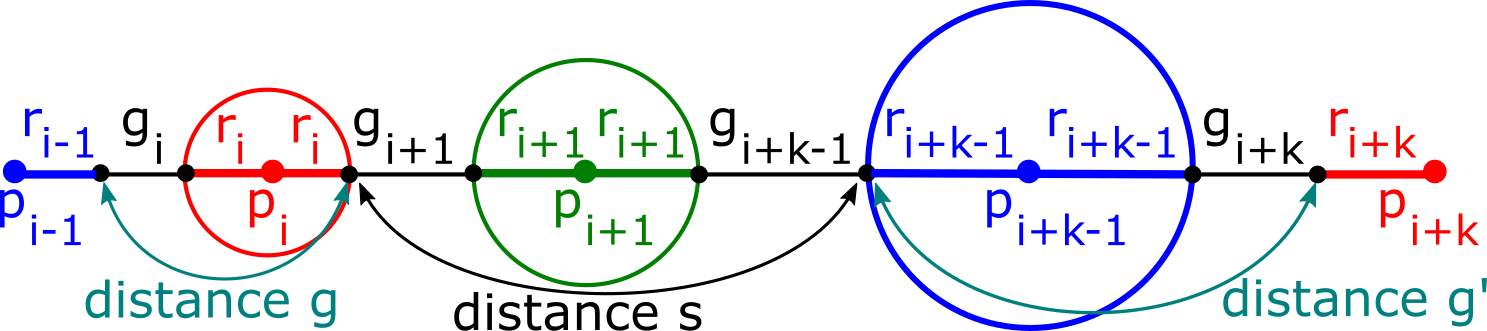}
\caption{The distances $g,s,g'$ between line intervals used in the proofs of Theorems~\ref{thm:1st_density} and~\ref{thm:k-th_density}, shown here for $k=3$. }
\label{fig:intervals_gaps}      
\end{figure*}

Example~\ref{exa:revisit_1st_density} applies Theorem~\ref{thm:1st_density} to get $\psi_1$ found for the periodic sequence $S$ in Example~\ref{exa:1st_density}.

\begin{exa}[using Theorem~\ref{thm:1st_density} for $\psi_1$]
\label{exa:revisit_1st_density}
The sequence $S=\left\{0,\dfrac{1}{3},\dfrac{1}{2}\right\}+\Z$ in Example~\ref{exa:1st_density} with points $p_1=0$, $p_2=\dfrac{1}{3}$, $p_3=\dfrac{1}{2}$ of radii $r_1=\dfrac{1}{12}$, $r_2=0$, $r_3=\dfrac{1}{12}$, respectively, has 
 the initial gaps between successive intervals 
$g_1=\dfrac{1}{3}$, 
$g_2=\dfrac{1}{4}$,
$g_3=\dfrac{1}{12}$, see all the computations in Example~\ref{exa:revisit_0-th_density}.
\medskip

\noindent
\textbf{Case (R)}.
In Theorem~\ref{thm:1st_density} for the trapezoid function $\eta_R=\eta_1$ measuring the fractional length covered only by the red interval, we set $k=1$ and $i=1$.
Then $r_i=\dfrac{1}{12}$, $g_i=\dfrac{1}{3}$ and $g_{i+1}=\dfrac{1}{4}$, so 
$$\dfrac{g_i+g_{i+1}}{2}+r_i=\dfrac{1}{2}\left(\dfrac{1}{3}+\dfrac{1}{4}\right)+\dfrac{1}{12}=\dfrac{3}{8},$$ 
$g=\min\{g_i,g_{i+1}\}=\dfrac{1}{4}$, $g+2r_i=\dfrac{1}{4}+\dfrac{2}{12}=\dfrac{5}{12}$.
\medskip

Then $\eta_{R}=\eta_{1}$ has the following corner points:
$$\begin{array}{l}
\left(0,2r_i\right)=\left(0,\dfrac{1}{6}\right), \quad
\left(\dfrac{g_{i}}{2},g+2r_i\right)=\left(\dfrac{1}{6},\dfrac{5}{12}\right), \\ \\
\left(\dfrac{g_{i+1}}{2},g+2r_i\right)=\left(\dfrac{1}{8},\dfrac{5}{12}\right), \\ \\ 
\left(\dfrac{g_{i}+g_{i+1}}{2}+r_i,0\right)=\left(\dfrac{3}{8},0\right),
\end{array}$$
where the two middle corners are accidentally swapped due to $g_i>g_{i+1}$ but they define the same trapezoid function as in the first picture of Fig.~\ref{fig:1st_density}.
\medskip

\noindent
\textbf{Case (G)}.
In Theorem~\ref{thm:1st_density} for the trapezoid function $\eta_G=\eta_2$ measuring the fractional length covered only by the green interval, we set $k=1$ and $i=2$.
Then $r_i=0$, $g_i=\dfrac{1}{4}$ and $g_{i+1}=\dfrac{1}{12}$, so 
$$\dfrac{g_i+g_{i+1}}{2}+r_i=\dfrac{1}{2}\left(\dfrac{1}{4}+\dfrac{1}{12}\right)+0=\dfrac{1}{6},$$ 
$g=\min\{g_i,g_{i+1}\}=\dfrac{1}{12}$, $g+2r_i=\dfrac{1}{12}+0=\dfrac{1}{12}$.
\medskip

Then $\eta_{G}=\eta_{2}$ has the following corner points
exactly as shown in the second picture of Fig.~\ref{fig:1st_density}:
$$\begin{array}{l}
\left(0,2r_i\right)=\left(0,0\right), \quad
\left(\dfrac{g_{i}}{2},g+2r_i\right)=\left(\dfrac{1}{8},\dfrac{1}{12}\right), \\ \\
\left(\dfrac{g_{i+1}}{2},g+2r_i\right)=\left(\dfrac{1}{24},\dfrac{5}{12}\right), \\ \\ 
\left(\dfrac{g_{i}+g_{i+1}}{2}+r_i,0\right)=\left(\dfrac{1}{6},0\right).
\end{array}$$

\noindent
\textbf{Case (B)}.
In Theorem~\ref{thm:1st_density} for the trapezoid function $\eta_B=\eta_3$ measuring the fractional length covered only by the blue interval, we set $k=1$ and $i=3$.
Then $r_i=\dfrac{1}{12}$, $g_i=\dfrac{1}{12}$ and $g_{i+1}=\dfrac{1}{3}$, so 
$$\dfrac{g_i+g_{i+1}}{2}+r_i=\dfrac{1}{2}\left(\dfrac{1}{12}+\dfrac{1}{3}\right)+\dfrac{1}{12}=\dfrac{7}{24},$$ 
$g=\min\{g_i,g_{i+1}\}=\dfrac{1}{12}$, $g+2r_i=\dfrac{1}{12}+\dfrac{2}{12}=\dfrac{1}{4}$.
\medskip

Then $\eta_{B}=\eta_{3}$ has the following corner points:
$$\begin{array}{l}
\left(0,2r_i\right)=\left(0,\dfrac{1}{6}\right), \quad
\left(\dfrac{g_{i}}{2},g+2r_i\right)=\left(\dfrac{1}{24},\dfrac{1}{4}\right), \\ \\
\left(\dfrac{g_{i+1}}{2},g+2r_i\right)=\left(\dfrac{1}{6},\dfrac{1}{4}\right), \\ \\ 
\left(\dfrac{g_{i}+g_{i+1}}{2}+r_i,0\right)=\left(\dfrac{7}{24},0\right)
\end{array}$$
exactly as shown in the third picture of Fig.~\ref{fig:1st_density}.
\bs
\end{exa}

\section{Higher density functions $\psi_k$}
\label{sec:k-th_density}

This section proves Theorem~\ref{thm:k-th_density} describing the $k$-th density function $\psi_k[S](t)$ for any $k\geq 2$ and a periodic sequence $S$ of disjoint intervals.
\medskip

To prepare the proof of Theorem~\ref{thm:k-th_density}, 
Example~\ref{exa:2nd_density} computes $\psi_2[S]$ for $S$ from Example~\ref{exa:0-th_density}.

\begin{exa}[$\psi_2$ for $S=\left\{0,\dfrac{1}{3},\dfrac{1}{2}\right\}+\Z$]
\label{exa:2nd_density}
The density $\psi_2(t)$ can be found as the sum of the \emph{trapezoid} functions $\eta_{GB},\eta_{BR},\eta_{RG}$, each measuring the length of a double intersection, see Fig.~\ref{fig:growing_intervals}.
\medskip

For the green interval $[\dfrac{1}{3}-t,\dfrac{1}{3}+t]$ and the blue interval $[\dfrac{5}{12}-t,\dfrac{7}{12}+t]$, the graph of the function $\eta_{GB}(t)$ is piecewise linear and starts at the point $(\dfrac{1}{24},0)$ because these intervals touch at $t=\dfrac{1}{24}$.
\medskip

The green-blue intersection $[\dfrac{5}{12}-t,\dfrac{1}{3}+t]$ grows until $t=\dfrac{1}{6}$, when the resulting interval $[\dfrac{1}{4},\dfrac{1}{2}]$ touches the red interval on the left.
At the same time, the graph of $\eta_{GB}(t)$ is linearly growing (with gradient 2) to the corner $(\dfrac{1}{6},\dfrac{1}{4})$, see Fig,~\ref{fig:2nd_density}.
\medskip

\begin{figure}[h!]
\includegraphics[width=\linewidth]{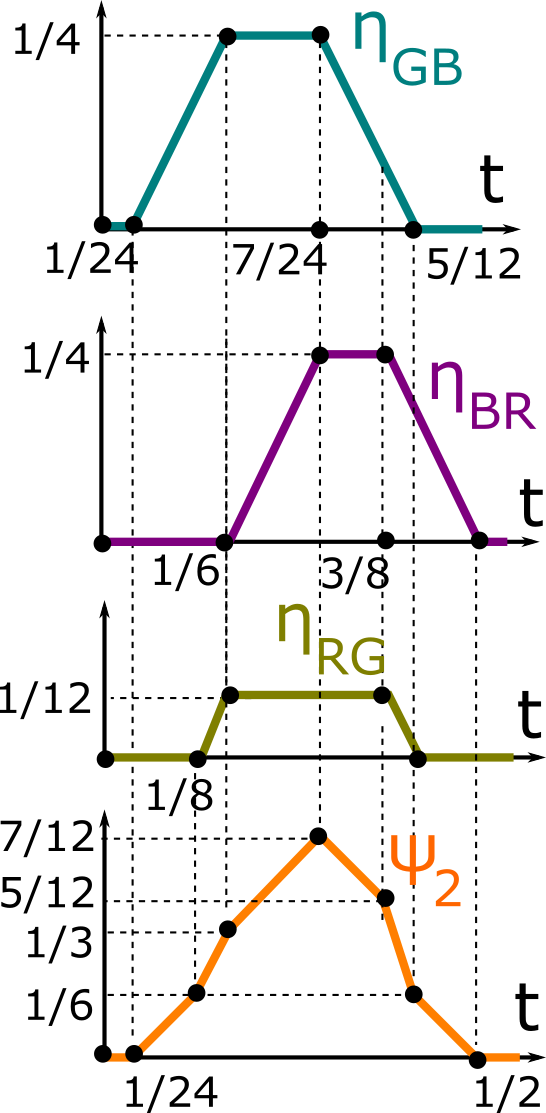}
\caption{The trapezoid functions $\eta_{GB},\eta_{BR},\eta_{RG}$ and the 2nd density function $\psi_2(t)$ for the 1-period sequence $S$ whose points $0,\dfrac{1}{3},\dfrac{1}{2}$ have radii $\dfrac{1}{12},0,\dfrac{1}{12}$, see Example~\ref{exa:2nd_density}. }
\label{fig:2nd_density}      
\end{figure}

For $t\in[\dfrac{1}{6},\dfrac{7}{24}]$, the green-blue intersection interval becomes shorter on the left, but grows at the same rate on the right until $t=\dfrac{7}{24}$ when $[\dfrac{1}{8},\dfrac{5}{8}]$ touches the red interval $[\dfrac{5}{8},1]$ on the right, see the 5th picture in Fig.~\ref{fig:growing_intervals}.
So the graph of $\eta_{GB}(t)$ remains constant up to the point $(\dfrac{7}{24},\dfrac{1}{4})$.
\medskip

For $t\in[\dfrac{7}{24},\dfrac{5}{12}]$ the green-blue intersection interval is shortening from both sides.
So the graph of $\eta_{GB}(t)$ linearly decreases (with gradient $-2$) and returns to the $t$-axis at the corner $(\dfrac{5}{12},0)$, then remains constant $\eta_{GB}(t)=0$ for $t\geq \dfrac{5}{12}$.
\medskip

Fig.~\ref{fig:2nd_density} shows all trapezoid functions for double intersections and $\psi_2=\eta_{GB}+\eta_{BR}+\eta_{RG}$.
\bs
\end{exa}

\begin{thm}[description of $\psi_k$ for $k\geq 2$]
\label{thm:k-th_density}
Let a periodic sequence $S=\{p_1,\dots,p_m\}+\Z$ consist of disjoint intervals with centers $0\leq p_1<\dots<p_m<1$ and radii $r_1,\dots,r_m\geq 0$, respectively. 
Consider the \emph{gaps} $g_i=(p_{i}-r_{i})-(p_{i-1}+r_{i-1})$ between the successive intervals of $S$, where $i=1,\dots,m$ and $p_{0}=p_m-1$, $r_0=r_m$. 
\medskip

For $k\geq 2$, the density function $\psi_k(t)$ equals the sum of $m$ \emph{trapezoid} functions $\eta_{k,i}(t)$, $i=1,\dots,m$, each having the following corner points: 
$$\left(\dfrac{s}{2},0\right),  
\left(\dfrac{g+s}{2},g\right),
\left(\dfrac{s+g'}{2},g\right),
\left(\dfrac{g+s+g'}{2},0\right),$$
where 
$g,g'$ are the minimum and maximum values in the pair $\{g_{i}+2r_i,g_{i+k}+2r_{i+k-1}\}$, and $s=\sum\limits_{j=i+1}^{i+k-1}g_j+2\sum\limits_{j=i+1}^{i+k-2} r_j$,
 so $s=g_{i+1}$ for $k=2$.
\medskip

Hence $\psi_k(t)$ is determined by the unordered set of the ordered tuples $(g,s,g')$, $i=1,\dots,m$.
\bs 
\end{thm}
\begin{proof}
The $k$-th density function $\psi_k(t)$ measures the total fractional length of $k$-fold intersections among $m$ intervals $L_i(t)=[p_i-r_i-t,p_i+r_i+t]$, $i=1,\dots,m$.
Now we visualize all such intervals $L_i(t)$ in the line $\R$ without mapping them modulo 1 to the unit cell $[0,1]$.
\medskip

Since all radii $r_i\geq 0$, only $k$ successive intervals can contribute to $k$-fold intersections.
So a $k$-fold intersection of growing intervals emerges only when two intervals $L_i(t)$ and $L_{i+k-1}(t)$ overlap because their intersection should be also covered by all the intermediate intervals $L_i(t),L_{i+1}(t),\dots,L_{i+k-1}(t)$.
\medskip

Then the density $\psi_k(t)$ equals the sum of the $m$ \emph{trapezoid functions} $\eta_{k,i}$, $i=1,\dots,m$, each equal to the length of the
$k$-fold intersection $\cap_{j=i}^{i+k-1} L_j(t)$ not covered by other intervals.
Then $\eta_{k,i}(t)$ remains 0 until the first critical moment $t$ when
$2t$ equals the distance between the points $p_i+r_i$ and $p_{i+k-1}-r_{i+k-1}$ in $\R$, see Fig.~\ref{fig:intervals_gaps}, so
$2t=\sum\limits_{j=i+1}^{i+k-1}g_j+2\sum\limits_{j=i+1}^{i+k-2} r_j=s$.
Hence $t=\dfrac{s}{2}$ and $(\dfrac{s}{2},0)$ is the first corner point of 
$\eta_{k,i}(t)$.
\medskip

At $t=\dfrac{s}{2}$, the interval of the $k$-fold intersection $\cap_{j=i}^{i+k-1} L_j(t)$ starts expanding on both sides.
Hence $\eta_{k,i}(t)$ starts increasing (with gradient 2) until the $k$-fold intersection touches one of the neighboring intervals $L_{i-1}(t)$ or $L_{i+k}(t)$ on the left or on the right.
\medskip

The left interval $L_{i-1}(t)$ touches the $k$-fold intersection $\cap_{j=i}^{i+k-1} L_j(t)$ when
$2t$ equals the distance from $p_{i-1}+r_{i-1}$ (the right endpoint of $L_{i-1}$) to $p_{i+k-1}-r_{i+k-1}$ (the left endpoint of $L_{i+k-1}$), see Fig.~\ref{fig:intervals_gaps}, so
$$2t=\sum\limits_{j=i}^{i+k-1}g_j+2\sum\limits_{j=i}^{i+k-2} r_j=g_{i}+2r_i+s.$$

The right interval $L_{i+k-1}(t')$ touches the $k$-fold intersection $\cap_{j=i}^{i+k-1} L_j(t')$ when
$2t'$ equals the distance from $p_{i}+r_{i}$ (the right endpoint of $L_{i}$) to $p_{i+k}-r_{i+k}$ (the left endpoint of $L_{i+k}$), see Fig.~\ref{fig:intervals_gaps}, so 
$$2t'=\sum\limits_{j=i+1}^{i+k}g_j+2\sum\limits_{j=i+1}^{i+k-1} r_j=s+g_{i+k}+2r_{i+k-1}.$$

If (say) $g_{i}+2r_i=g<g'=g_{i+k}+2r_{i+k-1}$, the $k$-fold intersection $\cap_{j=i}^{i+k-1} L_j(t)$ first touches $L_{i-1}$ at the earlier moment $t$ before reaching $L_{i+k}(t')$ at the later moment $t'$.
At the earlier moment, $\eta_{k,i}(t)$ equals $2(t-\dfrac{s}{2})=g_i+2r_i=g$ and has the corner $(\dfrac{g+s}{2},g)$.
\medskip

After that, the $k$-fold intersection is shrinking on the left and is expanding at the same rate on the right.
So the function $\eta_{k,i}(t)=g$ remains constant until the $k$-fold intersection touches the right interval $L_{i+k}(t')$.
At this later moment $t'=\dfrac{s+g_{i+k}}{2}+r_{i+k-1}=g'$, $\eta_{k,i}(t')$ still equals $g$ and has the corner $(\dfrac{s+g'}{2},g)$.
\medskip

If $g_{i}+2r_i=g'>g=g_{i+k}+2r_{i+k-1}$, the growing intervals $L_{i-1}(t)$ and $L_{i+k-1}(t)$ touch the $k$-fold intersection $\cap_{j=i}^{i+k-1} L_j(t)$ in the opposite order.
However, the above arguments lead to the same corners $(\dfrac{g+s}{2},g)$ and $(\dfrac{s+g'}{2},g)$ of $\eta_{k,i}(t)$.
If $g=g'$, the two corners collapse to one corner in the graph of $\eta_{k,i}(t)$.
\medskip

The $k$-fold intersection $\cap_{j=i}^{i+k-1} L_j(t)$ becomes fully covered when the intervals $L_{i-1}(t),L_{i+k}(t)$.
At this moment, $2t$ equals the distance from $p_{i-1}+r_{i-1}$ (the right endpoint of $L_{i-1}$) to $p_{i+k}-r_{i+k}$ (the left endpoint of $L_{i+k}$), see Fig.~\ref{fig:intervals_gaps}, so 
$2t=\sum\limits_{j=i}^{i+k}g_j+2\sum\limits_{j=i}^{i+k-1} r_j=g_i+2r_i+s+g_{i+k}+2r_{i+k-1}=g+s+g'.$
The graph of $\eta_{k,i}(t)$ has the final corner $\left(\dfrac{g+s+g'}{2},0\right)$.
\end{proof}

Example~\ref{exa:revisit_k-th_density} applies Theorem~\ref{thm:k-th_density} to get $\psi_2$ found for the periodic sequence $S$ in Example~\ref{exa:0-th_density}.

\begin{exa}[using Theorem~\ref{thm:k-th_density} for $\psi_2$]
\label{exa:revisit_k-th_density}
The sequence $S=\left\{0,\dfrac{1}{3},\dfrac{1}{2}\right\}+\Z$ in Example~\ref{exa:1st_density} with points $p_1=0$, $p_2=\dfrac{1}{3}$, $p_3=\dfrac{1}{2}$ of radii $r_1=\dfrac{1}{12}$, $r_2=0$, $r_3=\dfrac{1}{12}$, respectively, has 
 the initial gaps 
$g_1=\dfrac{1}{3}$, 
$g_2=\dfrac{1}{4}$,
$g_3=\dfrac{1}{12}$, see Example~\ref{exa:revisit_0-th_density}.
\medskip

In Theorem~\ref{thm:k-th_density}, the 2nd density function $\psi_2[S](t)$ is expressed as a sum of the trapezoid functions computed via their corners below.
\medskip

\noindent
\textbf{Case (GB)}.
For the function $\eta_{GB}$ measuring the double intersections of the green and blue intervals centered at $p_2=p_i$ and $p_3=p_{i+k-1}$, we set $k=2$ and $i=2$.
Then we have the radii $r_i=0$ and $r_{i+1}=\dfrac{1}{12}$, the gaps $g_i=\dfrac{1}{4}$, $g_{i+1}=\dfrac{1}{12}$, $g_{i+2}=\dfrac{1}{3}$, and the sum $s=g_{i+1}=\dfrac{1}{12}$.
The pair 
$$\left\{g_i+2r_i,g_{i+2}+2r_{i+1}\right\}=\left\{\dfrac{1}{4}+0,\dfrac{1}{3}+\dfrac{2}{12}\right\}$$ has the minimum value $g=\dfrac{1}{4}$ and maximum value $g'=\dfrac{1}{2}$.
Then $\eta_{2,2}[S](t)=\eta_{GB}$ has the following corners as expected in the top picture of Fig.~\ref{fig:2nd_density}: 
$$\begin{array}{l}
\left(\dfrac{s}{2},0\right)=\left(\dfrac{1}{24},0\right),\\ \\
\left(\dfrac{g+s}{2},g\right)
=\left(\dfrac{1}{2}\Big(\dfrac{1}{4}+\dfrac{1}{12}\Big),\dfrac{1}{4}\right)=\left(\dfrac{1}{6},\dfrac{1}{4}\right), \\ \\
\left(\dfrac{s+g'}{2},g\right)
=\left(\dfrac{1}{2}\Big(\dfrac{1}{12}+\dfrac{1}{2}\Big),\dfrac{1}{4}\right)=\left(\dfrac{7}{24},\dfrac{1}{4}\right), \\ \\
\Big(\dfrac{g+s+g'}{2},0\Big)
=\big(\dfrac{1}{2}(\dfrac{1}{4}+\dfrac{1}{12}+\dfrac{1}{2}),0\big)=\Big(\dfrac{5}{12},0\Big).
\end{array}$$

\noindent
\textbf{Case (BR)}.
For the trapezoid function $\eta_{BR}$ measuring the double intersections of the blue and red intervals centered at $p_3=p_i$ and $p_1=p_{i+k-1}$, we set $k=2$ and $i=3$.
Then we have the radii $r_i=\dfrac{1}{12}=r_{i+1}$, the gaps $g_i=\dfrac{1}{12}$, $g_{i+1}=\dfrac{1}{3}$, $g_{i+2}=\dfrac{1}{4}$, and $s=g_{i+1}=\dfrac{1}{3}$.
The pair 
$$\left\{g_i+2r_i,g_{i+2}+2r_{i+1}\right\}=\left\{\dfrac{1}{12}+\frac{2}{12},\dfrac{1}{4}+\frac{2}{12}\right\}$$ has the minimum $g=\dfrac{1}{4}$ and maximum $g'=\dfrac{5}{12}$.
Then $\eta_{2,3}[S](t)=\eta_{BR}$ has the following corners as expected in the second picture of Fig.~\ref{fig:2nd_density}: 
$$\begin{array}{l}
\left(\dfrac{s}{2},0\right)=\left(\dfrac{1}{6},0\right),\\ \\
\left(\dfrac{g+s}{2},g\right)
=\left(\dfrac{1}{2}\Big(\dfrac{1}{4}+\dfrac{1}{3}\Big),\dfrac{1}{4}\right)=\left(\dfrac{7}{24},\dfrac{1}{4}\right), \\ \\
\left(\dfrac{s+g'}{2},g\right)
=\left(\dfrac{1}{2}\Big(\dfrac{1}{3}+\dfrac{5}{12}\Big),\dfrac{1}{4}\right)=\left(\dfrac{3}{8},\dfrac{1}{4}\right), \\ \\
\Big(\dfrac{g+s+g'}{2},0\Big)
=\Big(\dfrac{1}{2}(\dfrac{1}{4}+\dfrac{1}{3}+\dfrac{5}{12}),0\Big)=\Big(\dfrac{1}{2},0\Big).
\end{array}$$

\noindent
\textbf{Case (RG)}.
For the trapezoid function $\eta_{RG}$ measuring the double intersections of the red and green intervals centered at $p_1=p_i$ and $p_2=p_{i+k-1}$, we set $k=2$ and $i=1$.
Then we have the radii $r_i=\dfrac{1}{12}$ and $r_{i+1}=0$, the gaps $g_i=\dfrac{1}{3}$, $g_{i+1}=\dfrac{1}{4}$, $g_{i+2}=\dfrac{1}{12}$, and $s=g_{i+1}=\dfrac{1}{4}$.
The pair 
$$\left\{g_i+2r_i,g_{i+2}+2r_{i+1}\right\}=\left\{\dfrac{1}{3}+\frac{2}{12},\dfrac{1}{12}+0\right\}$$ has the minimum $g=\dfrac{1}{12}$ and maximum $g'=\dfrac{1}{2}$.
Then $\eta_{2,1}[S](t)=\eta_{RG}$ has the following corners: 
$$\begin{array}{l}
\left(\dfrac{s}{2},0\right)=\left(\dfrac{1}{8},0\right),\\ \\
\left(\dfrac{g+s}{2},g\right)
=\left(\dfrac{1}{2}\Big(\dfrac{1}{12}+\dfrac{1}{4}\Big),\dfrac{1}{12}\right)=\left(\dfrac{1}{6},\dfrac{1}{12}\right), \\ \\
\left(\dfrac{s+g'}{2},g\right)
=\left(\dfrac{1}{2}\Big(\dfrac{1}{4}+\dfrac{1}{2}\Big),\dfrac{1}{12}\right)=\left(\dfrac{3}{8},\dfrac{1}{12}\right), \\ \\
\Big(\dfrac{g+s+g'}{2},0\Big)
=\Big(\dfrac{1}{2}(\dfrac{1}{12}+\dfrac{1}{4}+\dfrac{1}{2}),0\Big)=\Big(\dfrac{5}{12},0\Big).
\end{array}$$
 as expected in the third picture of Fig.~\ref{fig:2nd_density}.
\bs
\end{exa}

\section{Properties of new densities}
\label{sec:properties}

This section proves the periodicity of the sequence $\psi_k$ with respect to the index $k\geq 0$ in Theorem~\ref{thm:periodicity}, which was a bit unexpected from original Definition~\ref{dfn:densities}.
We start with the simpler example for the familiar 3-point sequence in Fig.~\ref{fig:growing_intervals}.

\begin{figure*}[h!]
\includegraphics[width=\linewidth]{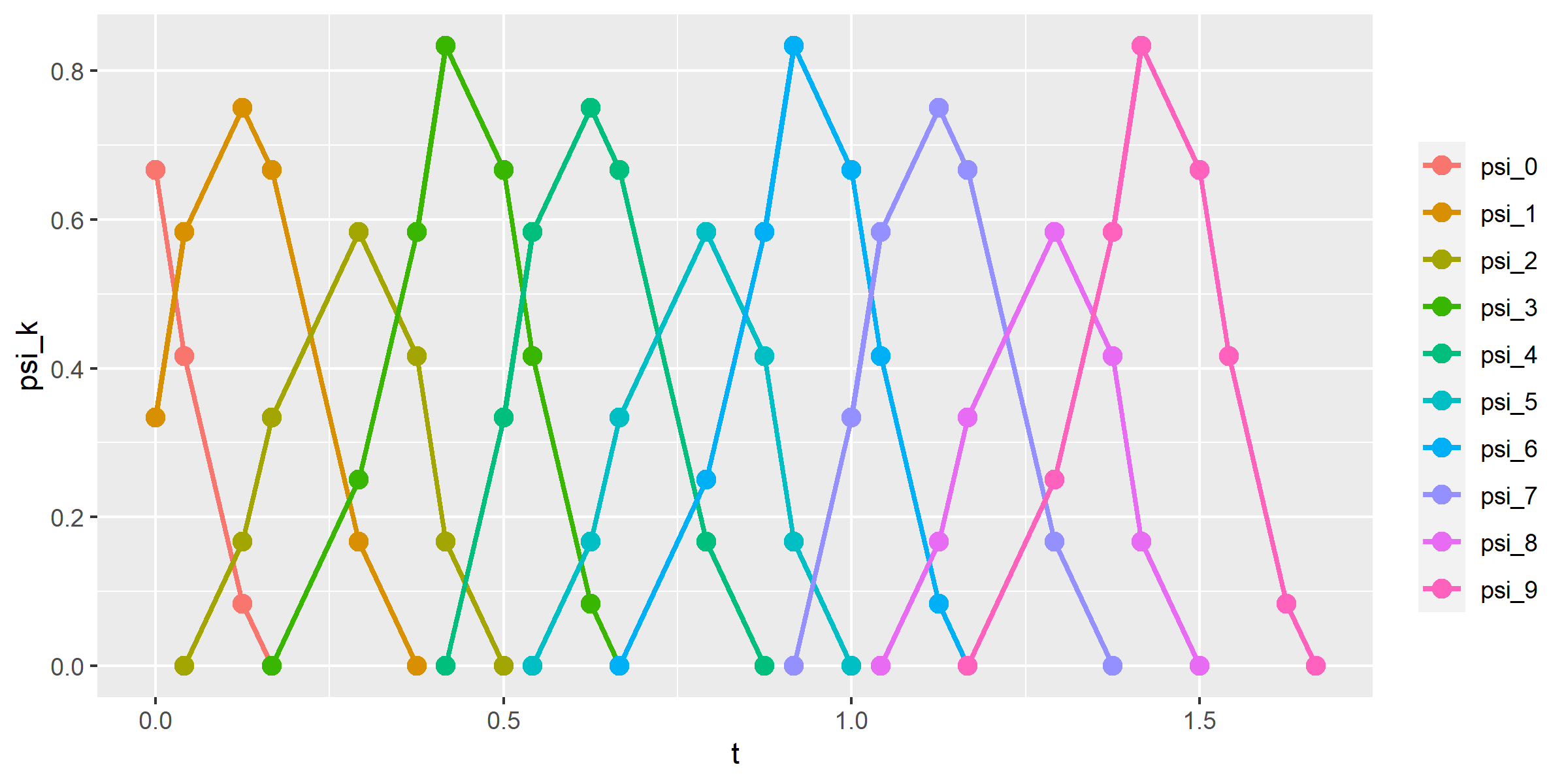}
\caption{The densities $\psi_k$, $k=0,\dots,9$ for the 1-period sequence $S$ whose points $0,\dfrac{1}{3},\dfrac{1}{2}$ have radii $\dfrac{1}{12},0,\dfrac{1}{12}$, respectively. 
The densities $\psi_0,\psi_1,\psi_2$ are described in Examples~\ref{exa:0-th_density},~\ref{exa:1st_density},~\ref{exa:2nd_density} and determine all other densities by periodicity in Theorem~\ref{thm:periodicity}. }
\label{fig:3-point_set_densities9}      
\end{figure*}

\begin{exa}[periodicity of $\psi_k$ in the index $k$]
\label{exa:periodicity}
Let the periodic sequence $S=\left\{0,\dfrac{1}{3},\dfrac{1}{2}\right\}+\Z$ have three points $p_1=0$, $p_2=\dfrac{1}{3}$, $p_3=\dfrac{1}{2}$ of radii $r_1=\dfrac{1}{12}$, $r_2=0$, $r_3=\dfrac{1}{12}$, respectively.
The initial intervals $L_1(0)=[-\frac{1}{12},\frac{1}{12}]$, $L_2(0)=[\frac{1}{3},\frac{1}{3}]$, $L_3(0)=[\frac{5}{12},\frac{7}{12}]$ have the 0-fold intersection measured by $\psi_0(0)=\dfrac{2}{3}$ and the 1-fold intersection measured by $\psi_1(0)=\dfrac{1}{3}$, see Fig.~\ref{fig:0-th_density} and~\ref{fig:1st_density}.
\medskip

By the time $t=\dfrac{1}{2}$ the initial intervals will grow to 
$L_1(\frac{1}{2})=[-\frac{7}{12},\frac{7}{12}]$, 
$L_2(\frac{1}{2})=[-\frac{1}{6},\frac{5}{6}]$, 
$L_3(\frac{1}{2})=[-\frac{1}{12},\frac{13}{12}]$.
The grown intervals at the radius $t=\dfrac{1}{2}$ have the 3-fold intersection $[-\frac{1}{12},\frac{7}{12}]$ of the length $\psi_3(\frac{1}{2})=\dfrac{2}{3}$, which coincides with $\psi_0(0)=\dfrac{2}{3}$.
\medskip

With the extra interval $L_4(\frac{1}{2})=[\frac{5}{12},\frac{19}{12}]$ centered at $p_4=1$, the 4-fold intersection is $L_1\cap L_2\cap L_3\cap L_4=[\frac{5}{12},\frac{7}{12}]$.
With the extra interval $L_{5}(\frac{1}{2})=[\frac{5}{6},\frac{11}{6}]$ centered at $p_{5}=\dfrac{4}{3}$, the 4-fold intersection $L_2\cap L_3\cap L_4\cap L_5$ is the single point $\dfrac{5}{6}$.
With the extra interval $L_{6}(\frac{1}{2})=[\frac{11}{12},\frac{13}{12}]$ centered at $p_6=\dfrac{3}{2}$, the 4-fold intersection is 
$L_3\cap L_4\cap L_5\cap L_6=[\frac{11}{12},\frac{13}{12}]$.
Hence the total length of the 4-fold intersection at $t=\dfrac{1}{2}$ is $\psi_4(\frac{1}{2})=\frac{1}{3}$, which coincides with $\psi_1(0)=\dfrac{1}{3}$.
\medskip

For the larger $t=1$, the six grown intervals
$$\begin{array}{ll}
L_1(1)=\left[-\dfrac{13}{12},\dfrac{13}{12}\right], & 
L_2(1)=\left[-\dfrac{2}{3},\dfrac{4}{3}\right], \\ \\
L_3(1)=\left[-\dfrac{7}{12},\dfrac{19}{12}\right], & 
L_4(1)=\left[-\dfrac{1}{12},\dfrac{25}{12}\right], \\ \\
L_5(1)=\left[\dfrac{1}{3},\dfrac{7}{3}\right], &
L_6(1)=\left[\dfrac{5}{12},\dfrac{31}{12}\right]
\end{array}$$

have the 6-fold intersection 
$\left[\dfrac{5}{12},\dfrac{13}{12}\right]$ of length $\psi_6(1)=\dfrac{2}{3}$ coinciding with $\psi_0(0)=\psi_3(\frac{1}{2})=\dfrac{2}{3}$.
\bs 
\end{exa}

Corollary~\ref{thm:periodicity} proves that the coincidences in Example~\ref{exa:periodicity} are not accidental. The periodicity of $\psi_k$ with respect to $k$ is illustrated by Fig.~\ref{fig:3-point_set_densities9}.

\begin{thm}[periodicity of $\psi_k$ in the index $k$]
\label{thm:periodicity}
The density functions $\psi_k[S]$ of a periodic sequence $S=\{p_1,\dots,p_m\}+\Z$ consist of disjoint intervals with centers $0\leq p_1<\dots<p_m<1$ and radii $r_1,\dots,r_m\geq 0$, respectively, satisfy the \emph{periodicity} 
$\psi_{k+m}(t+\frac{1}{2})=\psi_{k}(t)$ for any $k\geq 0$ and $t\geq 0$.
\bs
\end{thm}
\begin{proof}
Since the initial intervals are disjoint, for $k\geq 0$, any $(k+m)$-fold intersection involves $k+m$ successive intervals $L_i(t),\dots,L_{i+k+m-1}(t)$ centered around the points of $S$.
Then we can find an interval $[x,x+1]$ covering exactly $m$ of these initial intervals of $S$. 
\medskip

By collapsing $[x,x+1]$ to the point $x$, any $(k+m)$-fold intersection of $k+m$ intervals grown by a radius $r\geq\dfrac{1}{2}$ becomes a $k$-fold intersection of $k$ intervals grown by $t=r-\dfrac{1}{2}$.
Both $k$-fold and $(k+m)$-fold intersections within any unit cell have the same fractional length, so $\psi_{k+m}(t+\frac{1}{2})=\psi_{k}(t)$ for any $t\geq 0$.
\end{proof}

The symmetry $\psi_{m-k}(\frac{1}{2}-t)=\psi_k(t)$ for $k=0,\dots,[\frac{m}{2}]$, and $t\in[0,\frac{1}{2}]$ from \cite[Theorem~8]{anosova2022density} no longer holds for points with different radii.
For example, $\psi_1(t)\neq \psi_2(\frac{1}{2}-t)$ for the periodic sequence $S=\left\{0,\dfrac{1}{3},\dfrac{1}{2}\right\}+\Z$, see Fig.~\ref{fig:1st_density},~\ref{fig:2nd_density}.
If all points have the same radius $r$, \cite[Theorem~8]{anosova2022density} implies the symmetry after replacing $t$ by $t+2r$.

\medskip

The main results of \cite{anosova2022density} implied that all density functions cannot distinguish the non-isometric sequences 
$S_{15} = \{0,1,3,4,5,7,9,10,12\}+15\Z$ and $Q_{15} = \{0,1,3,4,6,8,9,12,14\}+15\Z$ of points with zero radii.
Example~\ref{exa:S15+Q15} shows that the densities for sequences with non-zero radii are strictly stronger and distinguish the sequences $S_{15}\not\cong Q_{15}$.

\begin{exa}[$\psi_k$ for $S_{15},Q_{15}$ with neighbor radii]
\label{exa:S15+Q15}
For any point $p$ in a periodic sequence $S\subset\R$, define its \emph{neighbor} radius as the half-distance to a closest neighbor of $p$ within the sequence $S$.
\smallskip

This choice of radii respects the isometry in the sense that periodic sequences $S,Q$ with zero-sized radii are isometric if and only if $S,Q$ with neighbor radii are isometric.
Fig.~\ref{fig:S5+Q15} shows that the densities $\psi_k$ for $k\geq 2$ distinguish the non-isometric sequences $S_{15}$ and $Q_{15}$ scaled down by factor 15 to the unit cell $[0,1]$, see 
Example~\ref{exa:SQ15}.
\bs
\end{exa}

\begin{figure*}[h!]
\includegraphics[width=\linewidth]{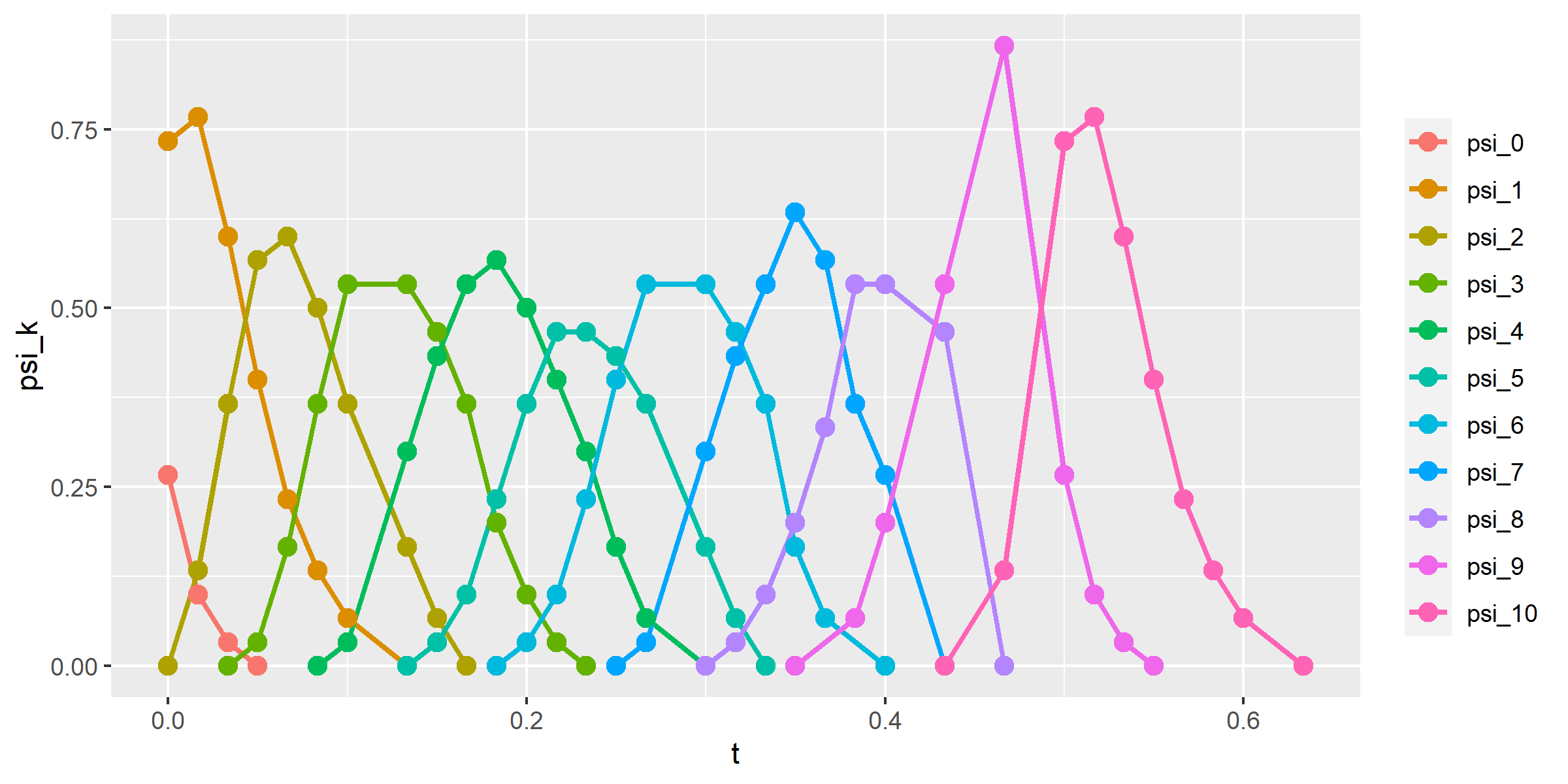}
\includegraphics[width=\linewidth]{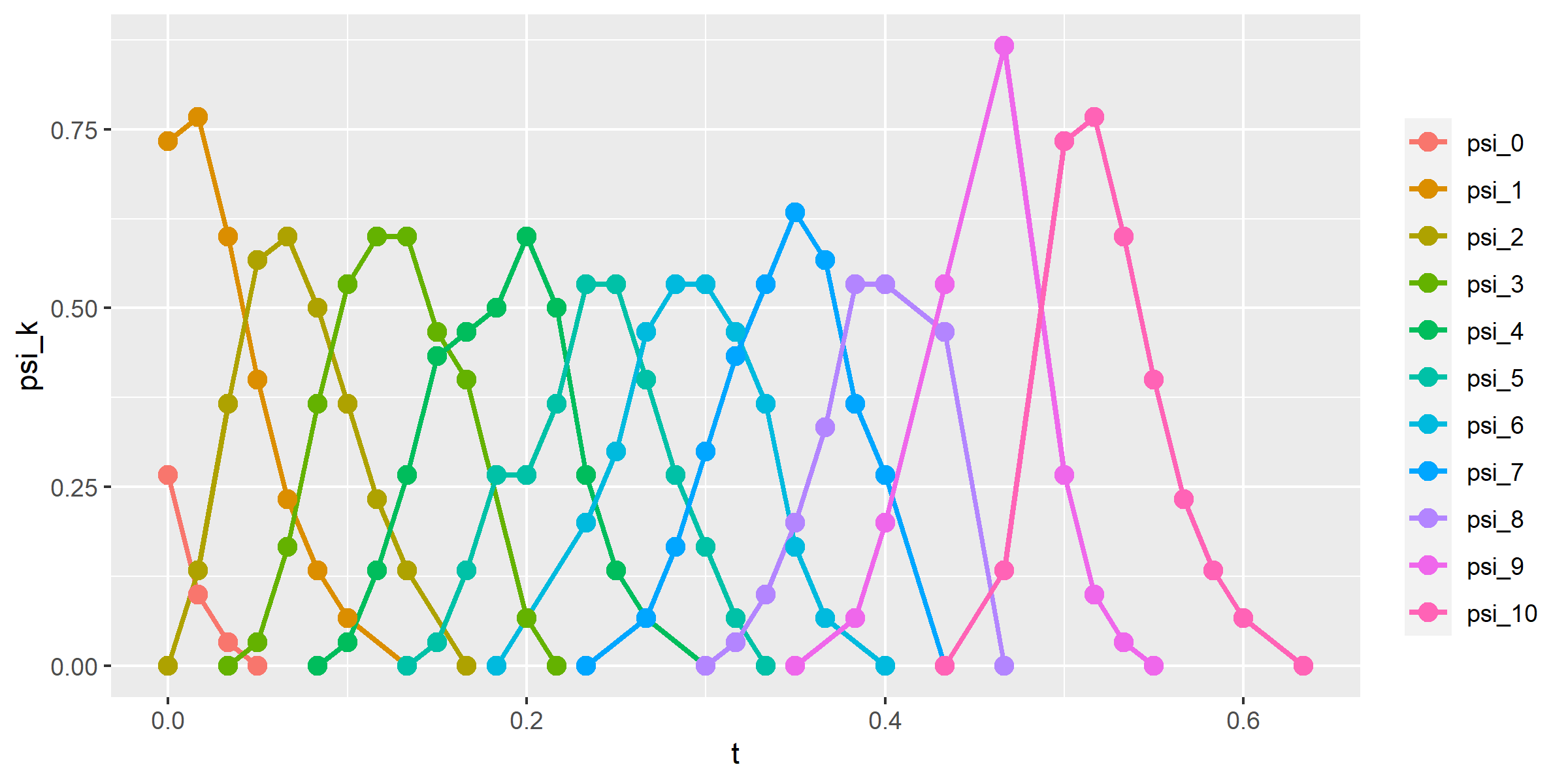}
\caption{The densities $\psi_k$, $k=0,\dots,10$, distinguish (already for $k\geq 2$) the sequences (scaled down by period $15$) $S_{15} = \{0,1,3,4,5,7,9,10,12\}+15\Z$ (\textbf{top}) and $Q_{15} = \{0,1,3,4,6,8,9,12,14\}+15\Z$ (\textbf{bottom}), where the radius $r_i$ of any point is the half-distance to its closest neighbor.
These sequences with zero radii have identical $\psi_k$ for all $k$, see \cite[Example~10]{anosova2022density}.}
\label{fig:S5+Q15}      
\end{figure*}

\begin{cor}[computation of $\psi_k(t)$]
\label{cor:computation}
Let $S,Q\subset\R$ be periodic sequences with at most $m$ motif points.
For $k\geq 1$, one can draw the graph of the $k$-th density function $\psi_k[S]$ in time $O(m^2)$.
One can check in time $O(m^3)$ if 
$\Psi[S]=\Psi[Q]$.
\bs
\end{cor}
\begin{proof}
To draw the graph of $\psi_k[S]$ or evaluate the $k$-th density function $\psi_k[S](t)$ at any radius $t$, we first use the periodicity from Theorem~\ref{thm:periodicity} to reduce $k$ to the range $0,1,\dots,m$.
In time $O(m\log m)$ we put the points from a unit cell $U$ (scaled to $[0,1]$ for convenience) in the increasing (cyclic) order $p_1,\dots,p_m$.
In time $O(m)$ we compute the gaps $g_i=(p_{i}-r_{i})-(p_{i-1}+r_{i-1})$ between successive intervals.
\smallskip

For $k=0$, we put the gaps in the increasing order $g_{[1]}\leq\dots\leq g_{[m]}$ in time $O(m\log m)$.
By Theorem~\ref{thm:0-th_density} in time $O(m^2)$, we write down the $O(m)$ corner points whose horizontal coordinates are the critical radii where $\psi_0(t)$ can change its gradient. 
\medskip

We evaluate $\psi_0$ at every critical radius $t$ by summing up the values of $m$ trapezoid functions at $t$, which needs $O(m^2)$ time.
It remains to plot the points at all $O(m)$ critical radii $t$ and connect the successive points by straight lines, so the total time is $O(m^2)$.
\medskip

For any larger fixed index $k=1,\dots,m$, in time $O(m^2)$ we write down all $O(m)$ corner points from Theorems~\ref{thm:1st_density} and~\ref{thm:k-th_density}, which leads to the graph of $\psi_k(t)$ similarly to the above argument for $k=0$.
\medskip

To decide if the infinite sequences of density functions coincide: $\Psi[S]=\Psi[Q]$, by Theorem~\ref{thm:periodicity} it suffices to check only if $O(m)$ density functions coincide: $\psi_k[S](t)=\psi_k[Q](t)$ for $k=0,1,\dots,[\frac{m}{2}]$.
\medskip

To check if two piecewise linear functions coincide, it remains to compare their values at all $O(m)$ critical radii $t$ from the corner points in Theorems~\ref{thm:0-th_density},~\ref{thm:1st_density},~\ref{thm:k-th_density}.
Since these values were found in time $O(m^2)$ above, the total time for $k=0,1,\dots,[\frac{m}{2}]$ is $O(m^3)$. 
\end{proof}

All previous examples show densities with a single local maximum.
However, the new R code \cite{anosova2023R} helped us discover the opposite examples.

\begin{figure}[h!]
\includegraphics[width=\linewidth]{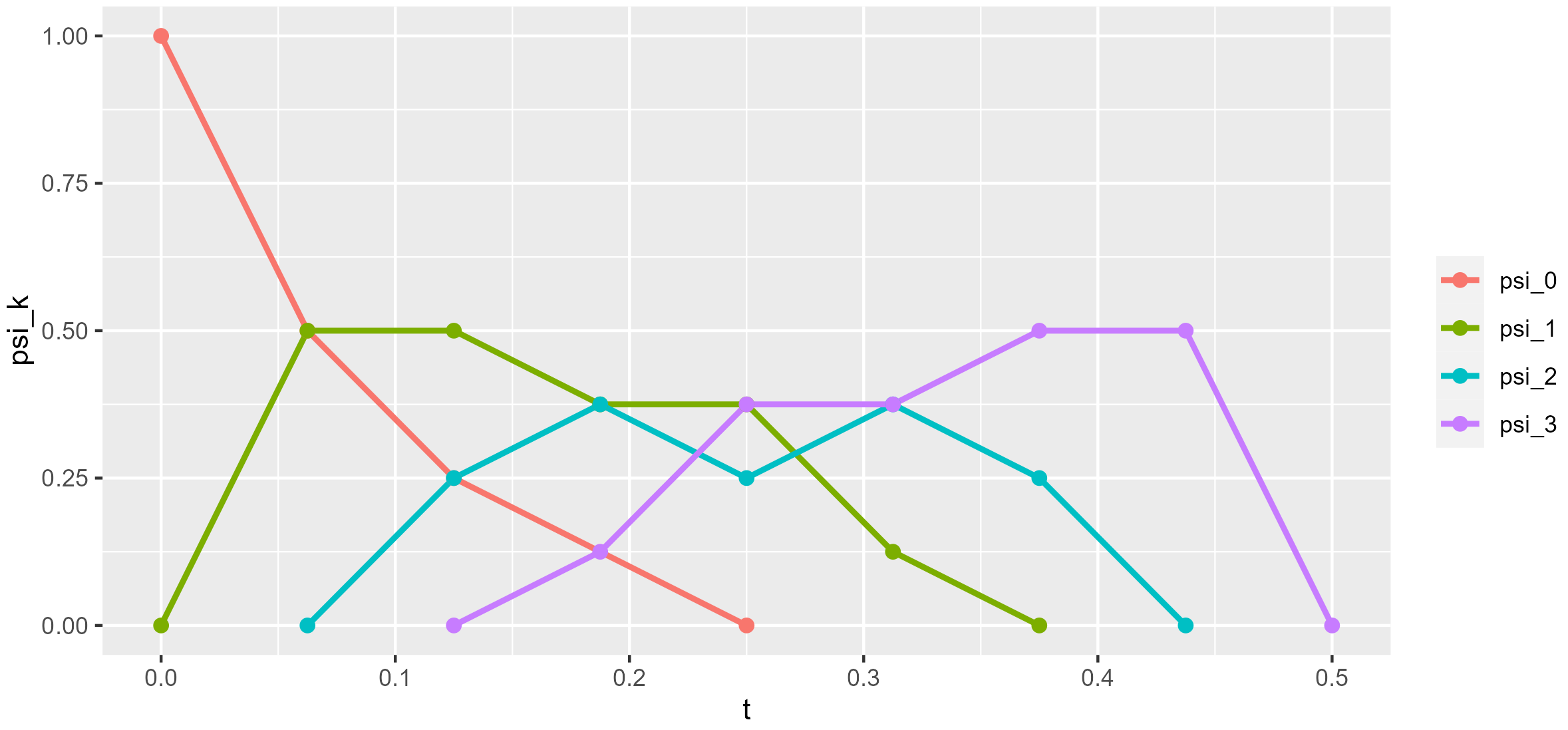}
\caption{For the periodic sequence $S=\left\{0,\dfrac{1}{8},\dfrac{1}{4},\dfrac{3}{4}\right\}+\Z$ whose all points have radii 0, the 2nd density $\psi_2[S](t)$ has the local minimum at $t=\dfrac{1}{4}$ between two local maxima.
}
\label{fig:set0_0125_025_075_densities3}      
\end{figure}

\begin{exa}[densities with multiple maxima]
\label{exa:multimax}
Fig.~\ref{fig:set0_0125_025_075_densities3} shows a simple 4-point sequence $S$ whose 2nd density $\psi_2[S]$ has two local maxima.
Fig.~\ref{fig:set0_powers3_psi_2_eta} and~\ref{fig:set0_powers2_psi_3_eta} show more complicated sequences whose density functions have more than two maxima.
\bs
\end{exa}

\begin{figure}[h!]
\includegraphics[width=\linewidth]{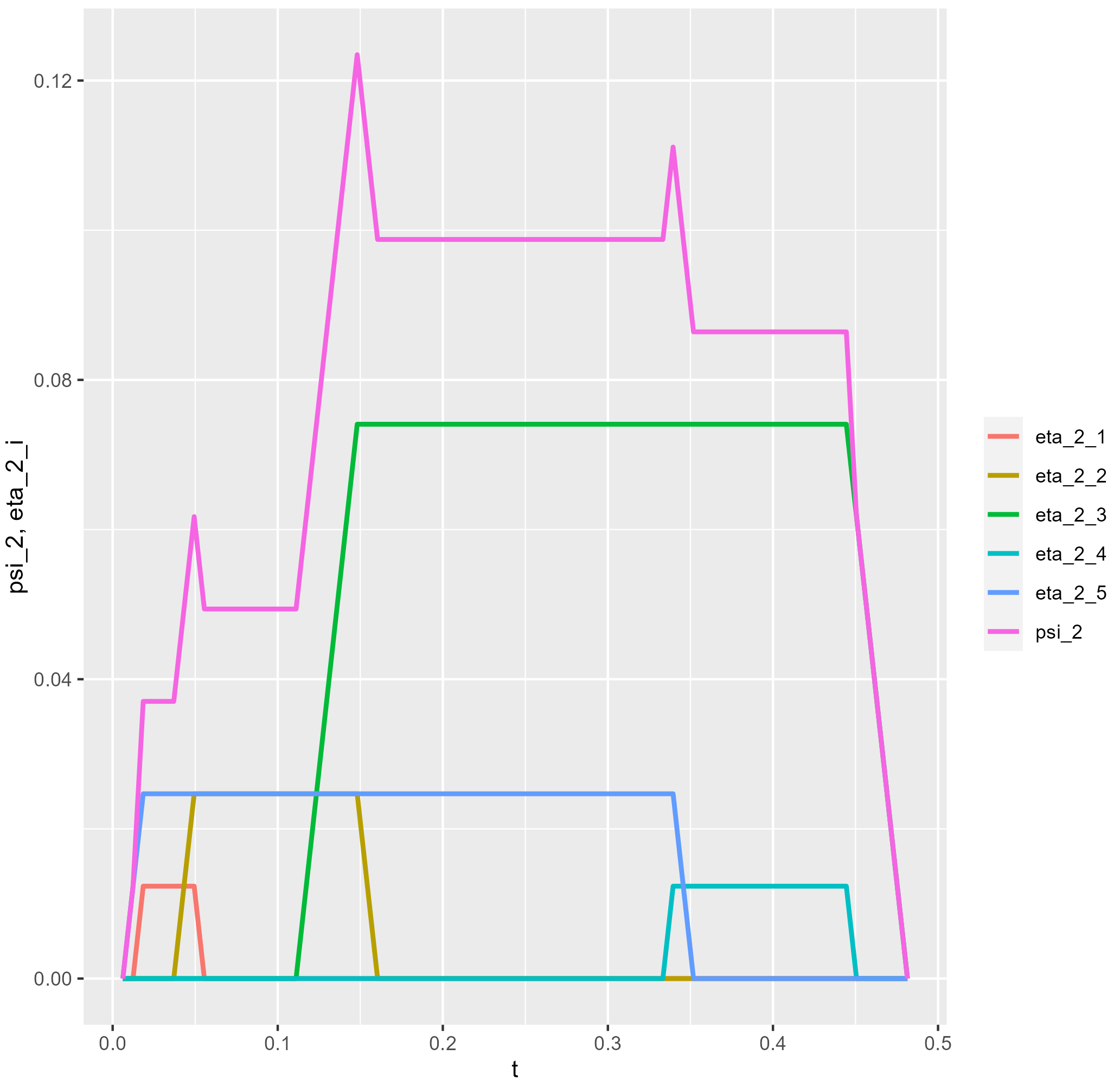}
\caption{For the sequence $S=\Big\{0,\dfrac{1}{81},\dfrac{1}{27},\dfrac{1}{9},\dfrac{1}{3}\Big\}+\Z$ whose all points have radii 0, $\psi_2[S]$ 
equal to the sum of the shown five trapezoid functions has three maxima.
}
\label{fig:set0_powers3_psi_2_eta}      
\end{figure}

\begin{figure}[h!]
\includegraphics[width=\linewidth]{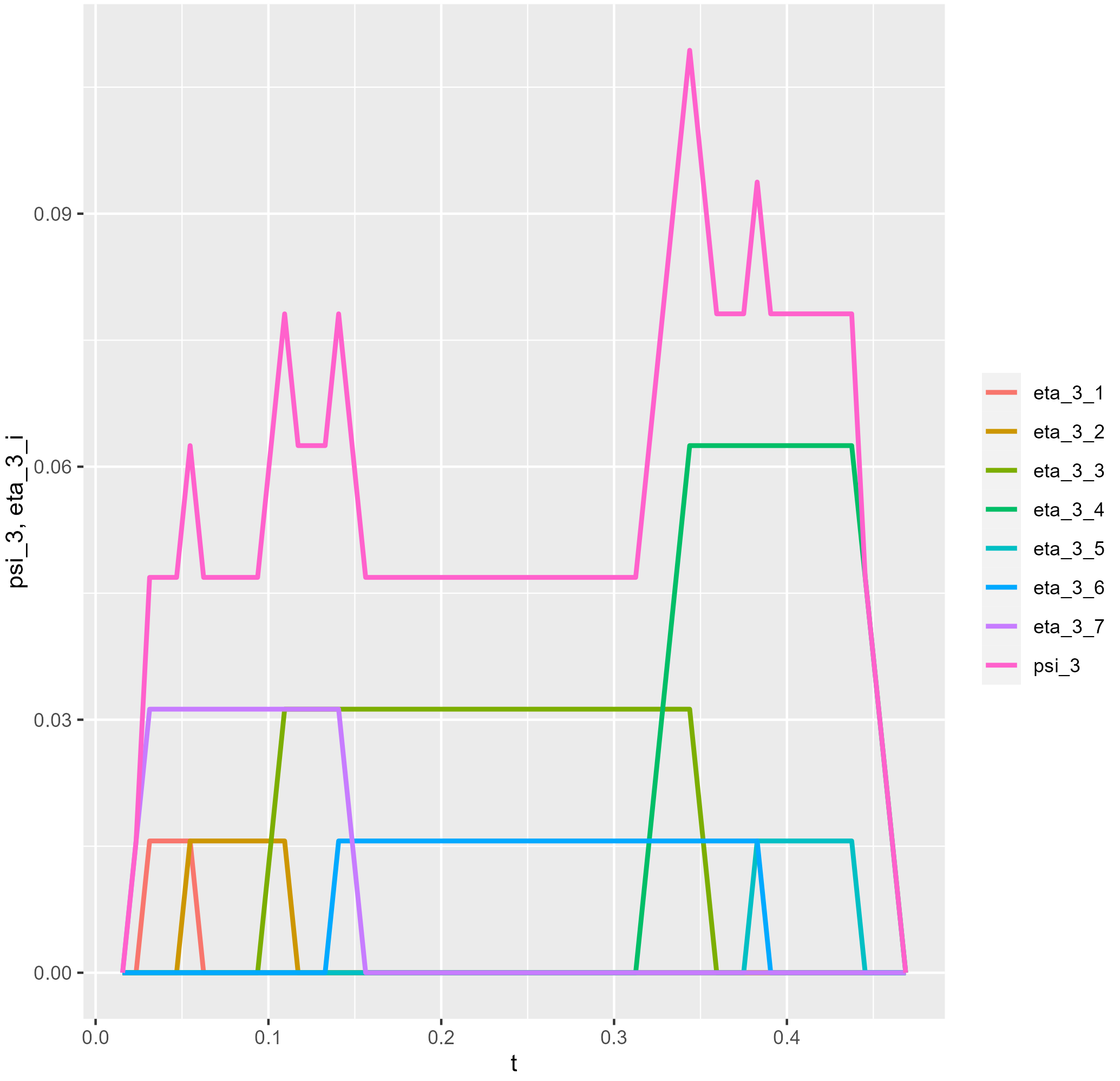}
\caption{For the sequence $S=\Big\{0,\dfrac{1}{64},\dfrac{1}{16},\dfrac{1}{8},\dfrac{1}{4},\dfrac{3}{4}\Big\}+\Z$ whose all points have radii 0, 
$\psi_3[S]$ has 5 local maxima.
}
\label{fig:set0_powers2_psi_3_eta}      
\end{figure}

\section{Conclusions and future work}
\label{sec:conclusions}

In comparison with the past work \cite{anosova2022density}, the key contributions of this paper are the following.
\medskip

\noindent
$\bullet$
Definition~\ref{dfn:densities} extends density functions $\psi_k$ to any periodic sets of points with radii $r_i\geq 0$.
\medskip

\noindent
$\bullet$
Theorems~\ref{thm:0-th_density},~\ref{thm:1st_density},~\ref{thm:k-th_density} explicitly  describe all $\psi_k$ for any periodic sequence $S$ of points with radii.
\medskip

\noindent
$\bullet$
The descriptions of $\psi_k$ allowed us to justify the periodicity of $\psi_k$ in Theorem~\ref{thm:periodicity} and a quadratic algorithm computing any $\psi_k$ in Corollary~\ref{cor:computation}.
\medskip

\noindent
$\bullet$
The code \cite{anosova2023R} helped us distinguish $S_{15}\not\cong Q_{15}$ in Example~\ref{exa:S15+Q15} and find sequences whose densities have multiple local maxima in Example~\ref{exa:multimax}. 
\medskip

Here are the open problems for future work.
\medskip

\noindent
$\bullet$
Verify if density functions $\psi_k[S](t)$ for small values of $k$ distinguish all non-isometric periodic point sets $S\subset\R^n$ at least with radii 0.
\medskip

\noindent
$\bullet$
Characterize the periodic sequences $S\subset\R$ whose all density functions $\psi_k$ for $k\geq 1$ have a unique local maximum, not as in Example~\ref{exa:multimax}.
\medskip

\noindent
$\bullet$
Similar to Theorems~\ref{thm:0-th_density},~\ref{thm:1st_density},~\ref{thm:k-th_density}, analytically describe the density function $\psi_k[S]$ for  periodic point sets $S\subset\R^n$ in higher dimensions $n>1$.
\medskip

\noindent
This research was supported by the grants of the UK Engineering Physical Sciences Research Council (EP/R018472/1, EP/X018474/1)  
and the Royal Academy of Engineering Industrial Fellowship (IF2122/186) of the last author.
We thank all reviewers 
for their time and helpful advice.

\bibliographystyle{splncs04}
\bibliography{densities-sequences-intervals}


\end{document}